%% file: main.tex
\documentclass[orivec,envcountsame,runningheads]{llncs}

\usepackage{centernot}
\usepackage{amssymb}
\usepackage{algorithm2e}
\usepackage{stmaryrd}
\usepackage{bm}
\usepackage{tikz}
\usetikzlibrary{shapes,calc,arrows,automata}
\usepackage{multirow}
\usepackage{array}
\usepackage{mathtools}
\usepackage{enumitem}
\usepackage[bookmarks,unicode,colorlinks=true]{hyperref}
\usepackage{proof}
\usepackage[capitalize,nameinlink]{cleveref}
\usepackage[location=appendix,manual]{moveproofs}
\usepackage{thm-restate}
\usepackage{subfig}
\usepackage{hhline}
\usepackage{graphicx}
\usepackage{microtype}
\usepackage{cite}
\usepackage{pgfplots}
\pgfplotsset{compat=1.18}
\usepackage{marvosym}

\DontPrintSemicolon

\hypersetup{
  pdftitle={ADCL: Acceleration Driven Clause Learning for Constrained Horn Clauses},
  colorlinks=true,
  linkcolor=blue,
  citecolor=olive,
  filecolor=magenta,
  urlcolor=cyan
}

\newenvironment{proofsketch}{%
  \proof}{\endproof}

\newcommand{\report}[1]{#1}
\newcommand{\submission}[1]{}

\renewcommand{\epsilon}{\varepsilon}
\let\oldphi\phi
\let\oldvarphi\varphi
\renewcommand{\varphi}{\oldphi}
\renewcommand{\phi}{\oldvarphi}

\newcommand{\pred}[1]{\mathsf{#1}}
\newcommand{\pl}[1]{\textsf{#1}}
\newcommand{\tool}[1]{\textsf{#1}}
\def\CC{{\pl{C}\nolinebreak[4]\hspace{-.05em}\raisebox{.4ex}{\tiny\bf ++}}}

\newcommand{\proj}[2]{\langle #1 \rangle_{#2}}
\newcommand{\mgu}{\mathsf{mgu}}
\newcommand{\ground}{\mathsf{grnd}}
\newcommand{\mbp}{\mathsf{sip}}
\newcommand{\unsat}{\mathsf{unsat}}
\newcommand{\unknown}{\mathsf{unknown}}
\newcommand{\sat}{\mathsf{sat}}
\newcommand{\bt}{\mathsf{bt}}
\newcommand{\cond}{\mathsf{cond}}
\newcommand{\accel}{\mathsf{accel}}
\newcommand{\resolve}{\mathsf{res}}
\newcommand{\F}{\pred{F}}
\newcommand{\G}{\pred{G}}
\renewcommand{\H}{\pred{H}}
\newcommand{\Inv}{\pred{Inv}}
\newcommand{\QF}{\mathsf{QF}}

\renewcommand{\AA}{\mathcal{A}}
\newcommand{\LL}{\mathcal{L}}
\newcommand{\VV}{\mathcal{V}}

\newcommand{\rs}{\leadsto_{\mathit{rs}}}

\newcommand{\NN}{\mathbb{N}}

\newcommand{\PP}{\mathcal{P}}

\newcommand{\Def}{\mathrel{\mathop:}=}
\newcommand{\relmiddle}[1]{\mathrel{}\middle#1\mathrel{}}

\renewcommand{\emptyset}{\varnothing}

\newcommand{\oldcomment}[1]{}

\newcommand{\anonymous}[2]{#1}

\DeclareMathOperator{\dom}{dom}

\crefname{equation}{eq.}{equations}%
\crefname{chapter}{chapter}{chapters}%
\crefname{section}{sect.}{sections}%
\crefname{appendix}{app.}{appendices}%
\crefname{enumi}{item}{items}%
\crefname{footnote}{footnote}{footnotes}%
\crefname{figure}{fig.}{figures}%
\crefname{table}{table}{tables}%
\crefname{theorem}{thm.}{theorems}%
\crefname{lemma}{lemma}{lemmas}%
\crefname{corollary}{cor.}{corollaries}%
\crefname{proposition}{proposition}{propositions}%
\crefname{definition}{def.}{definitions}%
\crefname{result}{result}{results}%
\crefname{example}{ex.}{examples}%
\crefname{remark}{remark}{remarks}%
\crefname{note}{note}{notes}%

\makeatletter
\RequirePackage[bookmarks,unicode,colorlinks=true]{hyperref}%
   \def\@citecolor{blue}%
   \def\@urlcolor{blue}%
   \def\@linkcolor{blue}%

\def\orcidID#1{\href{http://orcid.org/#1}{\smash{\protect\raisebox{-1.25pt}{\protect\includegraphics{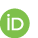}}}}}
\makeatother

\title{ADCL: Acceleration Driven Clause Learning for Constrained Horn Clauses\anonymous{\thanks{funded by
    the Deutsche Forschungsgemeinschaft (DFG, German Research Foundation)
    - 235950644 (Project GI 274/6-2)}}{}}
\author{\anonymous{Florian Frohn\submission{$^{(\href{mailto:florian.frohn@cs.rwth-aachen.de}{\mbox{\Letter}})}$}\orcidID{0000-0003-0902-1994}
    \and Jürgen Giesl\submission{$^{(\href{mailto:giesl@informatik.rwth-aachen.de}{\mbox{\Letter}})}$}\orcidID{0000-0003-0283-8520}}{}}
\authorrunning{Florian Frohn and Jürgen Giesl}
\institute{\anonymous{LuFG Informatik 2, RWTH Aachen University, Aachen,  Germany}{}}

\begin{document}

\maketitle

\input{abstract}
\submission{\vspace*{-.5cm}
\begin{center}
  \includegraphics[scale=0.16]{ValidatedBadge.png} \includegraphics[scale=0.16]{ExtensibleBadge.png} \includegraphics[scale=0.16]{AvailableBadge.png}%
\end{center}
\vspace*{-.5cm}}
\input{introduction}
\input{preliminaries}
\input{calculus}
\input{implementation}
\input{experiments}

\bibliographystyle{splncs04}
\input{main.bbl}

\report{\input{proofs}}

\end{document}

%% file: abstract.tex
\begin{abstract}
	\emph{Constrained Horn Clauses} (CHCs) are often used in automated program verification.
	Thus, techniques for (dis-)proving satisfiability of CHCs are a very active field of research.
	On the other hand, \emph{acceleration techniques} for computing formulas that characterize the $N$-fold closure of loops have successfully been used for static program analysis.
	We show how to use acceleration to avoid repeated derivations with recursive CHCs in resolution proofs, which reduces the length of the proofs drastically.
	This idea gives rise to a novel calculus for (dis)proving satisfiability of CHCs,
        called \emph{Acceleration Driven Clause Learning (ADCL)}.
        We implemented this new calculus in \anonymous{our}{the} tool \tool{LoAT} and evaluate it empirically in
        comparison to other state-of-the-art tools.
\end{abstract}

%% file: introduction.tex
\section{Introduction}
\label{sect:Introduction}

\emph{Constrained Horn Clauses} (CHCs) are often used for expressing
verification conditions in automated program verification.
Examples for tools based on CHCs include \tool{Korn} \cite{korn} and \tool{SeaHorn} \cite{seahorn} for verifying \pl{C} and \CC{} programs, \tool{JayHorn} for \pl{Java} programs \cite{jayhorn}, \tool{HornDroid} for \pl{Android} apps \cite{horndroid}, \tool{RustHorn} for \pl{Rust} programs \cite{rusthorn}, and \tool{SmartACE} \cite{smartACE} and \tool{SolCMC} \cite{solcmc} for \pl{Solidity}.
Consequently, techniques for \mbox{(dis-)}proving satisfiability of CHCs (CHC-SAT) are a very active field of research, resulting in powerful tools like \tool{Spacer} \cite{spacer}, \tool{Eldarica} \cite{eldarica}, \tool{FreqHorn} \cite{freqhorn}, \tool{Golem}
\cite{golem}, \tool{Ultimate} \cite{ultimate-chc}, and \tool{RInGEN} \cite{ringen}.

On the other hand, \emph{loop acceleration techniques} have been used successfully for static program analyses during the last years, resulting in tools like \tool{Flata} \cite{bozga10,iosif17} and \tool{LoAT} \cite{loat}.
Essentially, such techniques compute quantifier-free first-order formulas that characterize the $N$-fold closure of the transition relation of loops without branching in their body.
Thus, acceleration techniques can be used when generating verification conditions in order
to replace such loops with the closure of their transition relation.

In this paper, we apply acceleration techniques to CHC-SAT, where we restrict ourselves to linear CHCs, i.e., clauses that contain
at most one positive and one negative literal with uninterpreted predicates.
As our main interest lies in proving \emph{un}satisfiability of CHCs, our approach does not rely on abstractions, in contrast to most other techniques.
Instead, we use acceleration to cut off repeated derivations with recursive CHCs while exploring the state space via resolution.
In this way, the number of resolution steps that are required to reach a counterexample can be reduced drastically, as new CHCs that are ``learned'' via acceleration can simulate arbitrarily many ``ordinary'' resolution steps at once.
\begin{example}
  \label{ex:ex1}
  Consider the following set of CHCs $\Phi$ over the theory of linear integer arithmetic
  (LIA) with a \emph{fact}
$\phi_{\mathsf{f}}$, a \emph{rule} $\phi_{\mathsf{r}}$, and a \emph{query}
  $\phi_{\mathsf{q}}$, where $\top$ and $\bot$ stand for \emph{true} and \emph{false}:\footnote{{\tt chc-LIA-Lin\_052.smt2} from the benchmarks of the CHC
  Competition~'22 \cite{CHC-COMP}}
  \def\scale{0.93}
  \begin{align}
    \scalebox{\scale}{$\top$} & \scalebox{\scale}{${} \implies \Inv(0, 5000) \label{eq:ex1-fact} \tag{\protect{\ensuremath{\phi_{\mathsf{f}}}}}$} \\
    \scalebox{\scale}{$
    \begin{aligned}
      & \Inv(X_1, X_2) \land {}\\
      & \quad ((X_1 < 5000 \land Y_2 = X_2) \lor (X_1 \geq 5000 \land Y_2 = X_2 + 1))
    \end{aligned}$}
      & \scalebox{\scale}{${} \implies \Inv(X_1+1, Y_2) \label{eq:ex1-rule} \tag{\protect{\ensuremath{\phi_{\mathsf{r}}}}}$} \\
    \scalebox{\scale}{$\Inv(10000, 10000)$} & \scalebox{\scale}{${} \implies \bot \label{eq:ex1-query} \tag{\protect{\ensuremath{\phi_{\mathsf{q}}}}}$}
  \end{align}
  Its unsatisfiability can be proven via resolution and arithmetic simplification as shown in Fig.\ \ref{fig1}.
  \begin{figure}[t]
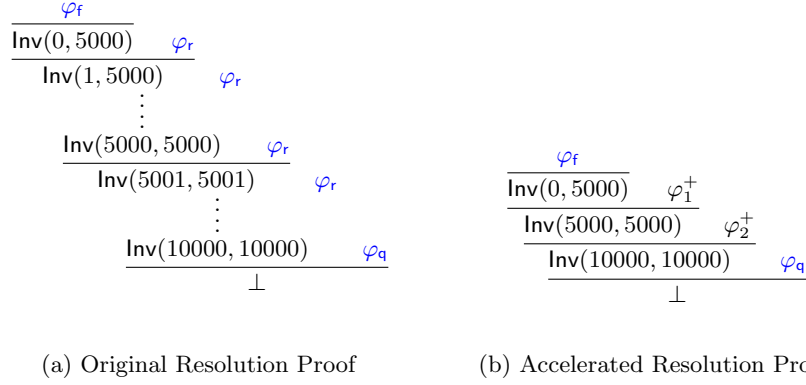

    \subfloat[Original Resolution Proof]{
      \begin{minipage}[b]{0.48\textwidth}
        \[
          \infer{\bot} {
            \infer*{\Inv(10000, 10000)} {
              \infer{\Inv(5001, 5001)}{
                \infer*{\Inv(5000, 5000)}{
                  \infer{\Inv(1, 5000)}{
                    \infer[]{\Inv(0, 5000)}{
                      \ref{eq:ex1-fact}
                    } & \ref{eq:ex1-rule}
                  } & \ref{eq:ex1-rule}
                } & \ref{eq:ex1-rule}
              } & \ref{eq:ex1-rule}
            } & \ref{eq:ex1-query}
          }
        \]
        \label{fig1}
      \end{minipage}
    }
    \subfloat[Accelerated Resolution Proof]{
      \begin{minipage}[b]{0.48\textwidth}
        \[
          \infer{\bot} {
            \infer{\Inv(10000, 10000)} {
              \infer{\Inv(5000, 5000)}{
                \infer[]{\Inv(0, 5000)}{
                  \ref{eq:ex1-fact}
                } & \phi^+_1
              } & \phi^+_2
            } & \ref{eq:ex1-query}
          }
        \]
        \label{fig2}
      \end{minipage}
    }
    \caption{Original and Accelerated Resolution Proof\vspace*{-.3cm}}
  \end{figure}
  The proof requires $10001$ resolution steps.
  Using acceleration techniques, we can derive the following two new CHCs from $\Phi$:
  \begin{alignat}{2}
    \Inv(X_1,X_2) & \land N>0 \land X_1 + N < 5001 && {} \implies \Inv(X_1+N,X_2) \tag{\protect{\ensuremath{\phi^+_1}}} \label{eq:accel1}\\
    \Inv(X_1,X_2) & \land N>0 \land X_1 \geq 5000 && {} \implies \Inv(X_1 + N,X_2 + N) \tag{\protect{\ensuremath{\phi^+_2}}} \label{eq:accel2}
  \end{alignat}
  The first CHC $\phi^+_1$ covers arbitrarily many subsequent resolution steps with \ref{eq:ex1-rule} where $X_1 < 5000$.
  Similarly, the second CHC $\phi^+_2$ covers arbitrarily many steps where $X_1 \geq 5000$.
  Now we can prove unsatisfiability of $\Phi$ with just 3 resolution steps, as shown in Fig.\ \ref{fig2}.
\end{example}

This idea gives rise to a novel calculus for CHC-SAT, called \emph{Acceleration Driven Clause Learning (ADCL)}.
ADCL is refutationally complete and can also prove satisfiability, but it does not necessarily terminate.

So far, our implementation in \anonymous{our}{the} tool \tool{LoAT} is restricted to proving unsatisfiability.
In program verification (which is one of the most important applications of CHC-SAT), satisfiability usually corresponds to safety, i.e., if an error state is reachable in the original program, then the CHCs derived from the program are unsatisfiable.
Hence, \tool{LoAT} can be used to show reachability of error states in program verification.
The ``witness'' of reachability is a resolution proof
ending in a conditional empty clause $\psi \implies \bot$ (where the condition $\psi$ is a formula over
the signature of some background theory like LIA),
together with a model for $\psi$.
Instantiating the variables in the proof according to the model yields a proof on ground instances.
Then this instantiated proof corresponds to a program run that ends in an error state, i.e., a counterexample.

After introducing preliminaries in \Cref{sec:preliminaries}, we formalize ADCL in \Cref{sec:ADCL}.
Next, we discuss how to implement ADCL efficiently in \Cref{sec:implementation}.
In \Cref{sec:experiments}, we discuss related work, and we show that our approach is highly competitive with state-of-the-art CHC-SAT solvers by means of an empirical evaluation.
\report{All proofs can be found in the appendix.}%
\submission{All proofs can be found in the extended version \cite{report}.}

%% file: preliminaries.tex
\section{Preliminaries}
\label{sec:preliminaries}

We assume that the reader is familiar with basic concepts from many-sorted first-order logic.
Throughout this paper, $\Sigma$ denotes a many-sorted first-order signature that just contains predicates, i.e., we do not consider uninterpreted functions.
Moreover, $\VV$ denotes a countably infinite set of variables and for each entity $e$, $\VV(e)$ denotes the variables occurring in $e$.
We write $\vec{x}$ for sequences and $x_i$ is the $i^{th}$ element of $\vec{x}$.
In the following, we introduce preliminaries regarding \emph{Constrained Horn Clauses} and \emph{acceleration techniques}.

\paragraph*{Constrained Horn Clauses} are first-order formulas of the form
\begin{align*}
  \forall \vec{X}_1,\ldots,\vec{X}_d,\vec{Y}, \vec{Z}.\ \F_1(\vec{X}_1) \land \ldots \land \F_d(\vec{X}_d) \land \psi & {} \implies \G(\vec{Y}) & \text{or} \\
  \forall \vec{X}_1,\ldots,\vec{X}_d, \vec{Z}.\ \F_1(\vec{X}_1) \land \ldots \land \F_d(\vec{X}_d) \land \psi         & {} \implies \bot
\end{align*}
where $\vec{X}_1,\ldots,\vec{X}_d,\vec{Y}, \vec{Z}$ are pairwise disjoint vectors of
pairwise different variables, $\F_1,\ldots,\F_d,\G \in \Sigma$, $\psi \in \QF(\AA)$ is a quantifier-free
first-order formula over the many-sorted signature $\Sigma_\AA$ of some theory $\AA$, and
$\VV(\psi) \subseteq \vec{X}_1 \cup\ldots \cup\vec{X}_d \cup\vec{Y} \cup
\vec{Z}$. We assume that $\AA$ is a complete theory with equality and
that $\Sigma$ and $\Sigma_\AA$ are disjoint, and we usually omit the leading universal quantifier, i.e., all variables in CHCs are implicitly universally quantified.\footnote{We assume that all arguments of predicates are
  variables.
  This is not a restriction, as one can add equations to $\psi$ to identify
$\Sigma_\AA$-terms with fresh variables.
  To ease the presentation, we also use $\Sigma_\AA$-terms as arguments of predicates in
  examples (e.g., in \Cref{ex:ex1} we wrote
  $\top \implies \Inv(0, 5000)$ instead of  $Y_1 = 0 \land Y_2 = 5000 \implies \Inv(Y_1, Y_2)$).}
Moreover, w.l.o.g., we assume that $\psi$ is in negation normal form.
In this paper, we restrict ourselves to \emph{linear} CHCs.
Thus, we consider CHCs of the following form:

\vspace*{-.15cm}

\noindent
\hspace*{-.2cm}\begin{minipage}{0.37\textwidth}
\begin{align}
  \label{eq:fact}
  \tag{fact}
  \psi                   & {} \implies \G(\vec{Y}) \\
  \label{eq:rule}
  \tag{rule}
  \F(\vec{X}) \land \psi & {} \implies \G(\vec{Y})
\end{align}
\end{minipage}
\hspace*{.1cm}
\begin{minipage}{0.62\textwidth}
    \begin{align}
  \label{eq:query}
  \tag{query}
  \F(\vec{X}) \land \psi & {} \implies \bot \\
  \label{eq:ground-fact}
  \tag{conditional empty clause}
  \psi                   & {} \implies \bot
\end{align}
\end{minipage}
\smallskip

\noindent
The premise and conclusion of a CHC is also called \emph{body} and \emph{head}, a CHC is \emph{recursive} if it is a rule where $\F = \G$, and it is \emph{conjunctive} if $\psi$ is a conjunction of literals.
Throughout this paper, $\phi$ and $\pi$ always denote CHCs.
The \emph{condition} of a CHC $\phi$ is $\cond(\phi) \Def \psi$.
We write $\phi|_{\psi'}$ for the CHC that results from $\phi$ by replacing $\cond(\phi)$ with $\psi'$.
Typically, the original set of CHCs does not contain conditional empty clauses, but in our
setting, such clauses result from resolution proofs that start with a fact and end with a query.
A conditional empty clause is called a \emph{refutation} if its condition is satisfiable.
We also refer to sets of CHCs as \emph{CHC problems}, denoted $\Phi$ or $\Pi$.

We call $\sigma$ an $\AA$-\emph{interpretation} if it is a model of $\AA$ whose carrier only contains ground terms over $\Sigma_\AA$, extended with interpretations for $\Sigma$ and $\VV$. 
Given a first-order formula $\eta$ over $\Sigma \cup \Sigma_\AA$, an $\AA$-interpretation $\sigma$ is a \emph{model} of $\eta$ (written $\sigma \models_\AA \eta$) if it satisfies $\eta$.
If such a model exists, then $\eta$ is \emph{satisfiable}.
As usual, $\models_\AA \eta$ means that $\eta$ is valid (i.e., we have $\sigma \models_\AA \eta$ for all $\AA$-interpretations $\sigma$) and $\eta \equiv_\AA \eta'$ means $\models_\AA \eta \iff \eta'$.
For sets of formulas $H$, we define $\sigma \models_\AA H$ if $\sigma \models_\AA
\bigwedge_{\eta \in H} \eta$.
The \emph{ground instances} of a CHC $\eta \land \psi \implies \eta'$ are:
\[
  \ground(\eta \land \psi \implies \eta') \Def \{\eta\sigma \implies \eta'\sigma \mid \sigma \models_\AA \psi\},
  \]
  where $\eta\sigma$ abbreviates $\sigma(\eta)$, i.e., it results from $\eta$ by instantiating all variables according to $\sigma$.
Since $\AA$ is complete (i.e., either $\models_\AA \psi$ or $\models_\AA \neg\psi$ holds
for every closed formula $\psi$ over $\Sigma_\AA$), $\mathcal{A}$-interpretations $\sigma$
only differ on $\Sigma$ and
$\mathcal{V}$,
and thus we have $\sigma \models_\AA \phi$ iff $\sigma \models_\AA \ground(\phi)$.

In the following, we use ``$::$'' for the concatenation of sequences, where we identify sequences of length $1$ with their elements, i.e., we sometimes write $x::xs$ instead of $[x]::xs$ or $x::y$ instead of $[x,y]$.
As usual, $\mgu(s,t)$ is the most general unifier of $s$ and $t$.
The following definition formalizes resolution (where we disregard the
underlying theory and just use ordinary syntactic unification). If the
corresponding literals of two clauses $\phi,\phi'$ do not unify, then we define their
resolvent to be $\bot \implies \bot$, so that resolution is defined for pairs of arbitrary CHCs.
Note that in our setting, the mgu $\theta$ is always a variable renaming.

\begin{definition}[Resolution]
  \label{def:resolution}
  Let $\phi$ and $\phi'$ be CHCs, where we rename the variables in $\phi$ and $\phi'$ such that they are disjoint.
  If
  \[
    \phi = (\eta \land \psi \implies \F(\vec{x})), \quad \phi' = (\F(\vec{y}) \land \psi' \implies \eta'), \quad \text{and} \quad \theta = \mgu(\F(\vec{x}),\F(\vec{y})),\vspace{-0.75em}
  \]
  \begin{align*}
    \text{then} && \resolve(\phi,\phi') & {} \Def (\eta \land \psi \land \psi' \implies \eta')\theta.\\
    \text{Otherwise,} && \resolve(\phi,\phi') & {} \Def (\bot \implies \bot).
  \end{align*}
  Here, $\eta$ can also be $\top$ and $\eta'$ can also be $\bot$.
  We lift $\resolve$ to non-empty sequences of CHCs by defining
  \[
    \resolve([\phi_1,\phi_2]::\vec{\phi}) \Def \resolve(\resolve(\phi_1,\phi_2)::\vec{\phi}) \qquad \text{and} \qquad \resolve(\phi_1) \Def \phi_1.
  \]
\end{definition}
We implicitly lift notations and terminology for CHCs to sequences of CHCs via resolution.
So for example, we have $\cond(\vec{\phi}) \Def \cond(\resolve(\vec{\phi}))$ and $\ground(\vec{\phi}) \Def \ground(\resolve(\vec{\phi}))$.

\begin{example}
  \label{ex:alternativeProof}
  Consider a variation $\Phi'$ of the CHC problem $\Phi$ from \Cref{ex:ex1} where \ref{eq:ex1-fact} is replaced by
  \begin{equation}
    \label{eq:ex1-alt-fact}
    \tag{\ensuremath{\phi'_{\mathsf{f}}}}
    X_1 \leq 0 \land X_2 \geq 5000 \implies \Inv(X_1,X_2)
  \end{equation}
  To prove its unsatisfiability, one can consider the resolvent of the sequence
  $\vec{\phi} \Def
  [\ref{eq:ex1-alt-fact},\ref{eq:accel1},\ref{eq:accel2},\ref{eq:ex1-query}]$:
  \begin{alignat*}{2}
    &&& \resolve([\ref{eq:ex1-alt-fact},\ref{eq:accel1},\ref{eq:accel2},\ref{eq:ex1-query}]) \\
    &{} = {} && \resolve([\psi \implies \Inv(X_1 + N,X_2),\ref{eq:accel2},\ref{eq:ex1-query}]) \\
    &{} = {} && \resolve([\psi \land N' > 0 \land X_1 {+} N \geq 5000 \implies \Inv(X_1 {+} N {+} N',X_2 {+} N'),\ref{eq:ex1-query}]) \tag{$\dagger$} \\
    &{} \equiv_\AA {} && \resolve([X_1 \leq 0 \land X_2 \geq 5000 \land N' > 0 \implies \Inv(5000+N',X_2+N'),\ref{eq:ex1-query}]) \tag{$\ddagger$} \\
    &{} = {} && (X_1 \leq 0 \land X_2 \geq 5000 \land N' > 0 \land 5000+N' = X_2+N' = 10000 \implies \bot)
  \end{alignat*}
  Here, we have
\[\psi \Def \cond([\ref{eq:ex1-alt-fact},\ref{eq:accel1}]) = X_1 \leq 0 \land X_2 \geq 5000 \land N > 0 \land X_1 + N < 5001.\]
  In the step marked with ($\dagger$), the variable $N'$ results from renaming $N$ in \ref{eq:accel2}.
  In the step marked with ($\ddagger$), we simplified $X_1 + N$ to $5000$ for readability, as $\psi \land X_1 + N \geq 5000$ implies $X_1 + N = 5000$.
  If $\sigma(X_1) = 0$ and $\sigma(X_2) = \sigma(N') = 5000$, then $\sigma \models_\AA \resolve(\vec{\phi})$ and thus $\resolve(\vec{\phi})$ is a refutation, so $\Phi$ is unsatisfiable.
  By instantiating the variables in the proof according to $\sigma$ (and setting $N$ to $5000$, as we had $X_1 + N = 5000$), we obtain an accelerated resolution proof on ground instances that is analogous to Fig.\ \ref{fig2} and serves as a ``witness'' of unsatisfiability, i.e., a ``counterexample'' to $\Phi'$.
\end{example}

\paragraph*{Acceleration Techniques} are used to compute the $N$-fold closure of the transition relation of a loop in program analysis.
In the context of CHCs, applying an acceleration technique to a recursive CHC $\phi$ yields another CHC $\phi'$ which, for any instantiation of a dedicated fresh variable $N \in \VV(\phi')$ with a positive integer, has the same ground instances as $\resolve(\phi^N)$.
Here, $\phi^N$ denotes the sequence consisting of $N$ repetitions of $\phi$.
In the following definition, we restrict ourselves to conjunctive CHCs, since many existing acceleration techniques do not support disjunctions \cite{bozga09a}, or have to resort to approximations in the presence of disjunctions \cite{acceleration-calculus}.

\begin{definition}[Acceleration]
  \label{def:accel}
  An \emph{acceleration technique} is a function $\accel$ that maps a recursive conjunctive CHC $\phi$ to a recursive conjunctive CHC such that $\ground(\accel(\phi)) = \bigcup_{n \in \NN_{\geq 1}} \ground(\phi^n)$.
\end{definition}

\begin{example}
  \label{ex:acceleration}
  In the CHC problem from \Cref{ex:ex1}, $\phi_{\mathsf{r}}$ entails
  \[
    \Inv(X_1, X_2) \land X_1 < 5000 \implies \Inv(X_1 + 1, X_2).
  \]
  From this CHC, an acceleration technique would compute \ref{eq:accel1}.
\end{example}

Note that most theories are not ``closed under acceleration''.
For example, consider the left clause below, which only uses linear arithmetic.
\[
  \F(X, Y) \implies \F(X + Y, Y) \qquad\qquad \F(X,Y) \land N > 0 \implies \F(X + N \cdot Y, Y)
\]
Accelerating it yields the clause on the right, which is not expressible with linear
arithmetic due to the sub-expression $N \cdot Y$.
Moreover, if there is no sort for integers in the background theory $\AA$, then an additional sort for the range of $N$ is required in the formula that results from applying $\accel$.
For that reason, we consider \emph{many-sorted} first-order logic and theories.

%% file: calculus.tex
\section{Acceleration Driven Clause Learning}
\label{sec:ADCL}

In this section, we introduce our novel calculus ADCL for (dis)proving satisfiability of CHC problems.
In \Cref{subsec:basics}, we start with
important concepts that ADCL is based on.
Then the ADCL calculus itself is presented in \Cref{subsec:adcl}.
Finally, in \Cref{subsec:properties}
we investigate the main properties of ADCL.

\subsection{Syntactic Implicants and Redundancy}
\label{subsec:basics}

Since ADCL relies on acceleration techniques, an important property of ADCL is that it only applies resolution to conjunctive CHCs, even if the analyzed CHC problem is not conjunctive.
To obtain conjunctive CHCs from non-conjunctive CHCs, we use \emph{syntactic implicants}.

\begin{definition}[Syntactic Implicant Projection]
  \label{def:implicant}
  Let $\psi \in \QF(\AA)$ be in negation normal form.
  We define:
  \begin{align*}
    \mbp(\psi,\sigma) & {} \Def \bigwedge \{\ell \text{ is a literal of } \psi \mid \sigma \models_\AA \ell\} && \text{if $\sigma \models_\AA \psi$} \\
    \mbp(\psi) & {} \Def \{ \mbp(\psi,\sigma) \mid \sigma \models_\AA \psi \}\\
    \mbp(\phi) & {} \Def \{\phi|_{\psi} \mid \psi \in \mbp(\cond(\phi))\} && \text{for CHCs $\phi$} \\
    \mbp(\Phi) & {} \Def \bigcup_{\phi \in \Phi} \mbp(\phi) && \text{for sets of CHCs $\Phi$}
  \end{align*}
  Here, $\mbp$ abbreviates \emph{syntactic implicant projection}.
\end{definition}
In contrast to the usual notion of implicants (which just requires that the implicants entail $\psi$), syntactic implicants are restricted to literals from $\psi$ to ensure that $\mbp(\psi)$ is finite.
We call such implicants \emph{syntactic} since \Cref{def:implicant} does not take the semantics of literals into account.
For example, the formula $\psi \Def (X > 0 \land X > 1)$ contains the literals $X > 0$ and $X > 1$, and $\models_\AA X > 1 \implies \psi$, but
$(X > 1) \notin \mbp(\psi) = \{\psi\}$,
because every model
of $X>1$ also satisfies $X>0$.
It is easy to show that $\psi \equiv_\AA \bigvee \mbp(\psi)$, and thus we also have $\Phi \equiv_\AA \mbp(\Phi)$.

\begin{example}\label{pip example}
  In the CHC problem of \Cref{ex:ex1}, we have
  \[
    \mbp(\cond(\ref{eq:ex1-rule})) = \{(X_1 < 5000 \land Y_2 = X_2), (X_1 \geq 5000 \land Y_2 = X_2 + 1)\}.
  \]
\end{example}

Since $\mbp(\phi)$ is worst-case exponential in the size of $\cond(\phi)$, we do not
compute it explicitly:
When resolving with $\phi$, we conjoin $\cond(\phi)$ to the condition of the resulting
resolvent and search for a model $\sigma$. 
This ensures that we do not continue with
resolvents that have unsatisfiable conditions.
Then we replace $\cond(\phi)$ by $\mbp(\cond(\phi),\sigma)$ in the resolvent.
This corresponds to a resolution step with
a conjunctive variant of $\phi$ whose condition is satisfied by $\sigma$.
In other words, our calculus constructs $\mbp(\cond(\phi), \sigma)$ ``on the fly'' when
resolving $\vec{\phi}$ with $\phi$, where $\sigma \models_\AA \cond(\vec{\phi}::\phi)$,
see \Cref{sec:implementation} for details.
In this way, the exponential blowup that results from constructing $\mbp(\phi)$ explicitly can often be avoided.

As ADCL learns new clauses via acceleration, it is important to prefer more general (learned) clauses over more specific clauses in resolution proofs.
To this end, we use the following \emph{redundancy relation} for CHCs.

\begin{definition}[Redundancy Relation]
  \label{def:redundancy}
 For two CHCs $\phi$ and $\pi$, we say that $\phi$ is \emph{(strictly) redundant} w.r.t.\ $\pi$, denoted $\phi \sqsubseteq \pi$ ($\phi \sqsubset \pi$), if $\ground(\phi) \subseteq \ground(\pi)$ ($\ground(\phi) \subset \ground(\pi)$).
  For a set of CHCs $\Pi$, we define $\phi \sqsubseteq \Pi$ ($\phi \sqsubset \Pi$) if $\phi \sqsubseteq \pi$
($\phi \sqsubset \pi$)
  for some $\pi \in \Pi$.
\end{definition}
In the following, we assume that we have oracles for checking redundancy, for satisfiability of $\QF(\AA)$-formulas, and for acceleration.
In practice, we have to resort to incomplete techniques instead.
In \Cref{sec:implementation}, we will explain how our implementation takes that into account.

\subsection{The ADCL Calculus}
\label{subsec:adcl}

A \emph{state} of ADCL consists of a CHC problem $\Pi$, containing the original CHCs and
all \emph{learned} clauses that were constructed by acceleration, the \emph{trace}
$[\phi_i]_{i=1}^k$,
representing a resolution proof, and a sequence $[B_i]_{i=0}^k$ of sets of \emph{blocking clauses}.
Clauses $\phi \sqsubseteq B_i$ must not be used for the $(i+1)^{th}$ resolution step.
In this way, blocking clauses prevent ADCL from visiting the same part of the search space more than once.
ADCL blocks a clause $\phi$ after proving that $\bot$ (and thus $\unsat$) cannot be
derived after adding $\phi$ to the current trace, or if
the current trace $\vec{\phi}::\vec{\phi}'$ ends with $\phi$ and
there is another ``more general''
trace $\vec{\phi}::\vec{\pi}$ such that $\vec{\phi}' \sqsubseteq \vec{\pi}$ and
$|\vec{\phi}'| \geq |\vec{\pi}|$, where one of the two relations is strict.
In the following, $\Phi$ denotes the original CHC problem whose satisfiability is analyzed with ADCL.

\begin{definition}[State]
  \label{def:state}
  A \emph{state} is a triple
  \[
    (\Pi,[\phi_i]_{i=1}^k,[B_i]_{i=0}^k)
  \]
  where $\Pi \supseteq \Phi$ is a CHC problem, $B_i \subseteq \mbp(\Pi)$ for each $0 \leq
  i \leq k$, and $[\phi_i]_{i=1}^k \in \mbp(\Pi)^*$.
  The clauses in $\mbp(\Phi)$ are called \emph{original clauses} and all clauses in $\mbp(\Pi)\setminus \mbp(\Phi)$ are called \emph{learned clauses}.
  A clause $\phi \sqsubseteq B_k$ is \emph{blocked}, and $\phi$ is \emph{active} if
  it is not blocked and $\cond([\phi_i]_{i=1}^{k}::\phi)$ is satisfiable.
  \end{definition}
Now we are ready to introduce our novel calculus.
\def\scale{0.9}
\def\width{0.86\textwidth}
\begin{definition}[ADCL]
  \label{def:calc}
  Let the ``backtrack function'' $\bt$ be defined as
  \[
    \bt(\Pi,[\phi_i]_{i=1}^k,[B_0,\ldots,B_k]) \Def (\Pi,[\phi_i]_{i=1}^{k-1},[B_0,\ldots,B_{k-1} \cup \{\phi_k\}]).
  \]
  Our calculus is defined by the following rules.
  \begin{align}
    & \scalebox{\scale}{\parbox{\width}{
      \[
        \infer{
          \Phi \leadsto (\Phi,[],[\emptyset])
      }{}
      \]
    }}
    \label{calc:init}
    \tag{\textsc{Init}}                                                                                                          \\
    & \scalebox{\scale}{\parbox{\width}{
      \[
        \infer{
          (\Pi,\vec{\phi},\vec{B}) \leadsto (\Pi,\vec{\phi}::\phi,\vec{B}::\emptyset)
        }{
          \phi \in \mbp(\Pi) \text{ is active}
        }
      \]
    }}
    \label{calc:step}
    \tag{\textsc{Step}}                                                                                                          \\
    & \scalebox{\scale}{\parbox{\width}{
      \[
        \infer{
          (\Pi,\vec{\phi}::\vec{\phi}^\circlearrowleft,\vec{B}::\vec{B}^\circlearrowleft) \leadsto (\Pi \cup \{\phi\},\vec{\phi}::\phi,\vec{B}::\{\phi\})
        }{
          \vec{\phi}^\circlearrowleft \text{ is recursive} & |\vec{\phi}^\circlearrowleft| = |\vec{B}^\circlearrowleft| & \accel(\vec{\phi}^\circlearrowleft) = \phi
        }
      \]
    }}
    \tag{\textsc{Accelerate}}
    \label{calc:accelerate}
    \\
    & \scalebox{\scale}{\parbox{\width}{
      \[
        \infer{
          s = (\Pi,\vec{\phi}::\vec{\phi}',\vec{B}) \leadsto \bt(s)
        }{
          \vec{\phi}' \sqsubset \mbp(\Pi) & \text{or} & \vec{\phi}' \sqsubseteq \mbp(\Pi) \land |\vec{\phi}'| > 1
        }
      \]
    }}
    \tag{\textsc{Covered}}
    \label{calc:covered}
    \\
    & \scalebox{\scale}{\parbox{\width}{
      \[
        \infer{
          s = (\Pi,\vec{\phi}::\phi,\vec{B}) \leadsto \bt(s)
        }{
          \text{all rules and queries from $\mbp(\Pi)$ are inactive} & \text{$\phi$ is not a query}
        }
      \]
    }}
    \tag{\textsc{Backtrack}}
    \label{calc:backtrack}
    \\
    & \scalebox{\scale}{\parbox{\width}{
      \[
        \infer{
          (\Pi,\vec{\phi},\vec{B}) \leadsto \unsat
        }{
          \vec{\phi} \text{ is a refutation}
        }
      \]
    }}
    \tag{\textsc{Refute}}
    \label{calc:refute}
    \\
    & \scalebox{\scale}{\parbox{\width}{
      \[
        \infer{
          (\Pi,[],[B]) \leadsto \sat
        }{
          \text{all facts and conditional empty clauses from $\mbp(\Pi)$ are inactive}
        }
      \]
    }}
    \tag{\textsc{Prove}}
    \label{calc:accept}
  \end{align}
\end{definition}

We write $\overset{\textsc{I}}{\leadsto}, \overset{\textsc{S}}{\leadsto}, \ldots$ to indicate that $\ref{calc:init}, \ref{calc:step}, \ldots$ was used for a $\leadsto$-step.
All derivations start with \ref{calc:init}.
\ref{calc:step} adds an active CHC $\phi$ to the trace.
Due to the linearity of CHCs, we can restrict ourselves to proofs that start with a fact
or a conditional empty clause, but such a restriction is not needed for the correctness of
our calculus and thus not enforced.

As soon as $\vec{\phi}$ has a recursive suffix $\vec{\phi}^\circlearrowleft$ (i.e., a suffix $\vec{\phi}^\circlearrowleft$ such that $\resolve(\vec{\phi}^\circlearrowleft)$ is recursive), \ref{calc:accelerate} can be used.
Then the suffix $\vec{\phi}^\circlearrowleft$ is replaced by the accelerated clause $\phi$
and the suffix $\vec{B}^\circlearrowleft$ of sets of blocked clauses that corresponds to
$\vec{\phi}^\circlearrowleft$ is replaced by $\{\phi\}$.
The reason is that for learned clauses, we always have $\resolve(\phi,\phi) \sqsubseteq \phi$, and thus applying $\phi$ twice in a row is superfluous.
So in this way, clauses that were just learned are not used for resolution several times in a row.
As mentioned in \Cref{sec:preliminaries}, the condition of the learned clause may not be expressible in the theory $\AA$.
Thus, when \ref{calc:accelerate} is applied, we may implicitly switch to a richer theory $\AA'$ (e.g., from linear to non-linear arithmetic).

If a suffix $\vec{\phi}'$ of the trace is redundant w.r.t.\ $\mbp(\Pi)$, we can backtrack
via \ref{calc:covered}, which removes the last element from $\vec{\phi}'$
  (but not the rest of $\vec{\phi}'$, since this sequence could now be continued in a different way) and blocks it, such that we do not revisit the corresponding part of the search space.
So here the redundancy check allows us to use more general (learned) clauses, if available.
Here, it is important that we do not backtrack if $\vec{\phi}'$ is a single, \emph{weakly} redundant clause.
Otherwise, \ref{calc:covered} could always be applied after
\ref{calc:step} or \ref{calc:accelerate} and block the last clause from the trace.
Thus, we might falsely ``prove'' satisfiability.

If no further \ref{calc:step} is possible since all CHCs are inactive, then we
\ref{calc:backtrack} as well and block the last clause from $\vec{\phi}$ to avoid
performing the same \ref{calc:step} again.

If we started with a fact and the last CHC in $\vec{\phi}$ is a query, then $\resolve(\vec{\phi})$ is a refutation and \ref{calc:refute} can be used to prove $\unsat$.

Finally, if we arrive in a state where $\vec{\phi}$ is empty and all facts and conditional empty clauses are inactive, then \ref{calc:accept} is applicable as we have exhausted the entire search space without proving unsatisfiability, i.e., $\Phi$ is satisfiable.
Note that we always have $|\vec{B}| = |\vec{\phi}| + 1$, since we need one additional set of blocking clauses to block facts.
While $B_0$ is initially empty (see \ref{calc:init}), it can be populated via \ref{calc:backtrack} or \ref{calc:covered}.
So eventually, all facts may become blocked, such that $\sat$ can be proven via \ref{calc:accept}.

\def\scale{0.99}
\begin{example}
  \label{ex:calc}
  Using our calculus, unsatisfiability of the CHC problem $\Phi$ in \Cref{ex:ex1} can be proven
  as follows:
 \begin{align*}
   \Phi \overset{\textsc{I}}{\leadsto} {} & (\Phi,[],[\emptyset]) \\
   {} \overset{\textsc{S}}{\leadsto} {} & (\Phi,[\ref{eq:ex1-fact}],[\emptyset,\emptyset]) \tag*{\scalebox{\scale}{$\top \implies \Inv(0,5000)$}} \\
   {} \overset{\textsc{S}}{\leadsto} {} & (\Phi,[\ref{eq:ex1-fact},\ref{eq:ex1-rule}|_{\psi_1}],[\emptyset,\emptyset,\emptyset]) \tag*{\scalebox{\scale}{$\Inv(0,5000) \implies \Inv(1,5000)$}} \\
   {} \overset{\textsc{A}}{\leadsto} {} & (\Pi_1,[\ref{eq:ex1-fact},\ref{eq:accel1}], [\emptyset,\emptyset,\{\ref{eq:accel1}\}]) \tag*{\scalebox{\scale}{$\Inv(0,5000) \implies \Inv(1,5000)$}} \\
   {} \overset{\textsc{S}}{\leadsto} {} & (\Pi_1,[\ref{eq:ex1-fact},\ref{eq:accel1},\ref{eq:ex1-rule}|_{\psi_2}], [\emptyset,\emptyset,\{\ref{eq:accel1},\},\emptyset]) \tag*{\scalebox{\scale}{$\Inv(5000,5000) \implies \Inv(5001,5001)$}} \\
   {} \overset{\textsc{A}}{\leadsto} {} & (\Pi_2,[\ref{eq:ex1-fact},\ref{eq:accel1},\ref{eq:accel2}], [\emptyset,\emptyset,\{\ref{eq:accel1}\},\{\ref{eq:accel2}\}]) \tag*{\scalebox{\scale}{$\Inv(5000,5000) {\implies} \Inv(5001,5001)$}} \\
   {} \overset{\textsc{S}}{\leadsto} {} & (\Pi_2,[\ref{eq:ex1-fact},\ref{eq:accel1},\ref{eq:accel2},\ref{eq:ex1-query}], [\emptyset,\emptyset,\{\ref{eq:accel1}\},\{\ref{eq:accel2}\},\emptyset]) \tag*{\scalebox{\scale}{$\Inv(10000,10000) \implies \bot$}} \\
   {} \overset{\textsc{R}}{\leadsto} {} & \unsat
 \end{align*}
 Here, we have:
 \begin{align*}
   \Pi_1 & {} \Def \Phi \cup \{\ref{eq:accel1}\}  & \psi_1 & {} \Def X_1 < 5000 \land Y_2 = X_2 \\
   \Pi_2 & {} \Def \Pi_1 \cup \{\ref{eq:accel2}\} & \psi_2 & {} \Def X_1 \geq 5000 \land Y_2 = X_2 + 1
  \end{align*}
  Beside the state of our calculus, we show a ground instance of the last element $\phi$
  of $\vec{\phi}$ which results from applying a model for $\cond(\vec{\phi})$ to
  $\phi$.
  In our implementation, we always maintain such a model.
  In general, these models are not unique:
  For example, after the first acceleration step, we might use $\Inv(0,5000) \implies \Inv(X_1,5000)$ for any $X_1 \in [1,5000]$.
  The reason is that $\ref{eq:accel1}$ can simulate arbitrarily many resolution steps with $\ref{eq:ex1-rule}|_{\psi_1}$, depending on the choice of $N$.

  After starting the derivation with \ref{calc:init}, we apply the only fact \ref{eq:ex1-fact} via \ref{calc:step}.
  Next, we apply \ref{eq:ex1-rule}, projected to the case $X_1 < 5000$.
  Since \ref{eq:ex1-rule} is recursive, we may apply \ref{calc:accelerate} afterwards, resulting in the new clause $\ref{eq:accel1}$.

  Then we apply \ref{eq:ex1-rule}, projected to the case $X_1 \geq 5000$.
  Note that the current model (resulting in the ground head-literal $\Inv(1,5000)$) cannot
  be extended to a model for $\ref{eq:ex1-rule}|_{\psi_2}$ (which requires $X_1 \geq 5000$).
  However, as the model is not part of the state, we may choose a different one at any
  point, which is important for implementing ADCL via \emph{incremental} SMT, see
  \Cref{sec:implementation}.
  Hence, we can apply $\ref{eq:ex1-rule}|_{\psi_2}$ nevertheless.

  Now we apply \ref{calc:accelerate} again, resulting in the new clause $\ref{eq:accel2}$.
  Finally, we apply the only query \ref{eq:ex1-query} via \ref{calc:step}, resulting in a conditional empty clause with a satisfiable condition, such that we can finish the proof via \ref{calc:refute}.

Later (in \Cref{def:reasonable}), we will define \emph{reasonable strategies} for applying the
rules of our calculus, which ensure that we use
\textsc{Accelerate} instead of applying
\textsc{Step} $10001$ times in our example.

  To see how our calculus proves satisfiability, assume that we replace \ref{eq:ex1-query} with
  \[
    \Inv(10000, X_2) \land X_2 \neq 10000 \implies \bot.
  \]
  Then resolution with \ref{eq:ex1-query} via the rule \ref{calc:step}
  is no longer applicable and our derivation continues as follows after the second application of \ref{calc:accelerate}:
  \begin{align*}
                   & (\Pi_2,[\ref{eq:ex1-fact},\ref{eq:accel1},\ref{eq:accel2}], [\emptyset,\emptyset,\{\ref{eq:accel1}\},\{\ref{eq:accel2}\}])               \\
    {} \overset{\textsc{B}}{\leadsto} {} & (\Pi_2,[\ref{eq:ex1-fact},\ref{eq:accel1}], [\emptyset,\emptyset,\{\ref{eq:accel1},\ref{eq:accel2}\}]) \\
    {} \overset{\textsc{B}}{\leadsto} {} & (\Pi_2,[\ref{eq:ex1-fact}], [\emptyset,\{\ref{eq:accel1}\}]) \\
    {} \overset{\textsc{B}}{\leadsto} {} & (\Pi_2,[], [\{\ref{eq:ex1-fact}\}]) \\
    {} \overset{\textsc{P}}{\leadsto} {} & \sat
  \end{align*}
  For all three \ref{calc:backtrack}-steps, \ref{eq:ex1-query} is clearly inactive, as adding it to $\vec{\phi}$ results in a resolvent with an unsatisfiable condition.
  The first \ref{calc:backtrack}-step is possible since $\ref{eq:ex1-rule}|_{\psi_1}$ and $\ref{eq:accel1}$ are inactive, as they require $X_1 < 5000$ for the first argument $X_1$ of their body-literal, but $\ref{eq:accel2}$ ensures $Y_1 > 5000$ for the first argument $Y_1$ of its head-literal.
  Moreover, $\ref{eq:ex1-rule}|_{\psi_2}$ and $\ref{eq:accel2}$ are blocked, as $\ref{eq:ex1-rule}|_{\psi_2} \sqsubseteq \ref{eq:accel2}$.
  The second \ref{calc:backtrack}-step is performed since $\ref{eq:ex1-rule}|_{\psi_1}$, $\ref{eq:ex1-rule}|_{\psi_2}$,
  $\ref{eq:accel1}$, and $\ref{eq:accel2}$ are blocked
  (as $\ref{eq:ex1-rule}|_{\psi_1} \sqsubseteq \ref{eq:accel1}$ and $\ref{eq:ex1-rule}|_{\psi_2} \sqsubseteq \ref{eq:accel2}$).
  The third \ref{calc:backtrack}-step is possible since $\ref{eq:ex1-rule}|_{\psi_1}$ and $\ref{eq:accel1}$ are blocked, and $\ref{eq:ex1-rule}|_{\psi_2}$ and $\ref{eq:accel2}$ cannot be applied without applying $\ref{eq:accel1}$ first, so they are inactive.
  Thus, we reach a state where the only fact \ref{eq:ex1-fact} is blocked and hence \ref{calc:accept} applies.

  To see an example for \ref{calc:covered}, assume that we apply $\ref{eq:ex1-rule}|_{\psi_1}$ \emph{twice} before using \ref{calc:accelerate}.
  Then the following derivation yields the trace that we obtained after the first acceleration step above:
  \begin{samepage}
  \begin{align*}
    & (\Phi,[\ref{eq:ex1-fact},\ref{eq:ex1-rule}|_{\psi_1}],[\emptyset,\emptyset,\emptyset]) \\
    {} \overset{\textsc{S}}{\leadsto} {} & (\Phi,[\ref{eq:ex1-fact},\ref{eq:ex1-rule}|_{\psi_1},\ref{eq:ex1-rule}|_{\psi_1}],[\emptyset,\emptyset,\emptyset,\emptyset]) \\
    {} \overset{\textsc{A}}{\leadsto} {} & (\Pi_1,[\ref{eq:ex1-fact},\ref{eq:ex1-rule}|_{\psi_1},\ref{eq:accel1}],[\emptyset,\emptyset,\emptyset,\{\ref{eq:accel1}\}]) \\
    {} \overset{\textsc{C}}{\leadsto} {} & (\Pi_1,[\ref{eq:ex1-fact},\ref{eq:ex1-rule}|_{\psi_1}],[\emptyset,\emptyset,\{\ref{eq:accel1}\}]) \tag{as $[\ref{eq:ex1-rule}|_{\psi_1},\ref{eq:accel1}] \sqsubset \ref{eq:accel1}$}\\
    {} \overset{\textsc{C}}{\leadsto} {} & (\Pi_1,[\ref{eq:ex1-fact}],[\emptyset,\{\ref{eq:ex1-rule}|_{\psi_1}\}]) \tag{as $\ref{eq:ex1-rule}|_{\psi_1} \sqsubset \ref{eq:accel1}$}\\
    {} \overset{\textsc{S}}{\leadsto} {} & (\Pi_1,[\ref{eq:ex1-fact},\ref{eq:accel1}],[\emptyset,\{\ref{eq:ex1-rule}|_{\psi_1}\},\emptyset]) \tag{$\dagger$} \label{eq:after-covered}
  \end{align*}
  \end{samepage}
\end{example}

As one can see in the example above, our calculus uses \emph{forward reasoning}, i.e., it starts with a fact and resolves it with rules until a query applies.
Alternatively, one could use \emph{backward reasoning} by starting with a query and resolving it with rules until a fact applies, as in logic programming.

Our calculus could easily be adapted for backward reasoning.
Then it would start resolving with a query and aim for resolving with a fact, while all
other aspects of the calculus would remain unchanged.
Such an adaption would be motivated by examples like
\begin{align*}
  \F(\ldots) \land \ldots & {} \implies \G(\ldots) \\
  \G(\ldots) \land \ldots & {} \implies \H(\ldots) \\
  \G(\ldots) \land \ldots & {} \implies \bot
\end{align*}
where $\H$ is the entry-point of a satisfiable sub-problem.
With forward reasoning, ADCL might spend lots of time on that sub-problem, whereas
unsatisfiability would be proven after just two steps with backward reasoning.
However, in our tests, backward reasoning did not help on any example.
Presumably, the reason is that the benchmark set from our evaluation does not contain examples with such a structure.
Thus, we did not pursue this approach any further.

\subsection{Properties of ADCL}
\label{subsec:properties}

In this section, we investigate the main properties of ADCL.
Most importantly, ADCL is sound.

\begin{restatable}[Soundness]{theorem}{soundness}
  \label{thm:soundness}
  If $\Phi \leadsto^* \sat$, then $\Phi$ is satisfiable.
  If $\Phi \leadsto^* \unsat$, then $\Phi$ is unsatisfiable.
\end{restatable}
\begin{proofsketch}$\!\!\!$\footnote{\report{To improve readability, we only present some proof sketches in
      the paper and refer to the appendix for all full proofs.}%
  \submission{All full proofs can be found in the extended version \cite{report}.}}
  For $\unsat$, we have $\Phi \leadsto^* (\Pi,\vec{\phi},\vec{B}) \leadsto \unsat$ where
  $\Pi \equiv_\AA \Phi$ and $\vec{\phi} \in \mbp(\Pi)^*$ is a refutation.
  For $\sat$, assume that $\Phi$ is unsatisfiable, but $\Phi \leadsto s = (\Phi,[],[\emptyset]) \leadsto^* (\Pi,[],[B]) = s' \leadsto \sat$.
  Then there is a refutation $\vec{\phi} \in \mbp(\Pi)^*$ that is minimal in the sense
  that $\phi_i \not\sqsubset \mbp(\Pi)$ for all $1 \leq i \leq |\vec{\phi}|$ and $\vec{\phi}' \not\sqsubseteq \mbp(\Pi)$ for all infixes $\vec{\phi}'$ of $\vec{\phi}$ whose length is at least $2$.
  We say that $\vec{\phi}$ is \emph{disabled} by a state $(\Pi',\vec{\phi}',\vec{B}')$ if
  $\vec{\phi}'$ has a prefix $[\phi'_i]_{i=1}^k$ such that $\phi_i \equiv_\AA \phi'_i$ for
  all $1 \leq i \leq k$ and $\phi_{k+1} \equiv_{\AA} \phi$ for some $\phi \in B'_k$.
  Then $\vec{\phi}$ is disabled by $s'$, but not by $s$.
  Let $s^{(i)}$ be the last state in the derivation $\Phi \leadsto^* \sat$ where $\vec{\phi}$ is enabled.
  Then $s^{(i)} \overset{\textsc{S}}{\leadsto} s^{(i+1)}$ would
  imply that $\vec{\phi}$ is enabled in $s^{(i+1)}$;
  $s^{(i)} \overset{\textsc{A}}{\leadsto} s^{(i+1)}$ would imply that two consecutive clauses in $\vec{\phi}$ are both equivalent to the newly learned clause, contradicting minimality of $\vec{\phi}$;
  $s^{(i)} \overset{\textsc{C}}{\leadsto} s^{(i+1)}$ would imply that the trace of $s^{(i)}$ is not minimal, which also contradicts minimality of $\vec{\phi}$;
  and $s^{(i)} \overset{\textsc{B}}{\leadsto} s^{(i+1)}$ would imply that an element of $\vec{\phi}$ is strictly redundant w.r.t.\ the last set of blocking clauses in $s^{(i)}$, which again contradicts minimality of $\vec{\phi}$.
  Hence, we derived a contradiction.
  \qed
\end{proofsketch}
\makeproof*{thm:soundness}{
  The following easy lemmas and definitions are needed to prove soundness of our approach.
  We first prove that adding learned clauses results in an equivalent CHC problem.

  \begin{lemma}[$\leadsto$ Preserves Equivalence]
    \label{lem:correct}
    If $(\Pi,\vec{\phi},\vec{B}) \leadsto^* (\Pi',\vec{\phi}',\vec{B}')$, then $\Pi \equiv_\AA \Pi'$.
  \end{lemma}
  \begin{proof}
    We only consider the case $(\Pi,\vec{\phi},\vec{B}) \leadsto (\Pi',\vec{\phi}',\vec{B}')$.
    Then the claim also follows for $(\Pi,\vec{\phi},\vec{B}) \leadsto^* (\Pi',\vec{\phi}',\vec{B}')$ by a straightforward induction.

    \ref{calc:accelerate} adds $\accel(\vec{\phi}^\circlearrowleft)$ to $\vec{\phi}$ and $\Pi$.
    Hence, we have to prove that $\sigma \models_\AA \Pi$ implies $\sigma \models_\AA \accel(\vec{\phi}^\circlearrowleft)$.
    We have:
    \begin{align*}
      & \sigma \models_\AA \Pi \\
      {} \iff {} & \sigma \models_\AA \mbp(\Pi) & \tag{since $\mbp(\Pi) \equiv_\AA \Pi$} \\
      {} \implies {} & \sigma \models_\AA \resolve(\vec{\phi}^\circlearrowleft) \tag{by soundness of resolution, as $\vec{\phi}^\circlearrowleft \in \mbp(\Pi)^*$} \\
      {} \iff {} & \sigma \models_\AA \ground(\vec{\phi}^\circlearrowleft) \tag{as $\sigma$'s carrier only contains ground terms over $\Sigma_\AA$} \\
      {} \iff {} & \sigma \models_\AA \bigcup_{n \in \NN_{\geq 1}} \ground((\vec{\phi}^\circlearrowleft)^n) \tag{by soundness of resolution} \\
      {} \iff {} & \sigma \models_\AA \ground(\accel(\vec{\phi}^\circlearrowleft)) \tag{by \Cref{def:accel}} \\
      {} \iff {} & \sigma \models_\AA \accel(\vec{\phi}^\circlearrowleft) \tag{as $\sigma$'s carrier only contains ground terms over $\Sigma_\AA$}
    \end{align*}
    All other rules do not modify $\Pi$. \qed
  \end{proof}
The following relation is needed to prove the soundness of our calculus.
  \begin{definition}[$\succ_\Pi$]
    Given a CHC problem $\Pi$, $\vec{\phi}_1,\vec{\phi}_2,\vec{\phi}_3 \in \mbp(\Pi)^*$, and $\phi \in \mbp(\Pi)$, we define
    \[
      \vec{\phi}_1::\vec{\phi}_2::\vec{\phi}_3 \succ_\Pi \vec{\phi}_1::\phi::\vec{\phi}_3 \qquad \text{if} \qquad \vec{\phi}_2 \sqsubset \phi \quad \text{or} \quad \vec{\phi}_2 \sqsubseteq \phi \land |\vec{\phi}_2| > 1.
    \]
  \end{definition}
  Note that the definition of $\succ_{\Pi}$ is closely related to the definition of \ref{calc:covered}.
  In particular, the trace $\vec{\phi} :: \vec{\phi}'$ is not minimal w.r.t.\ $\succ_{\Pi}$ whenever \ref{calc:covered} is applicable.
  \begin{lemma}
    \label{lem:wf}
    $\succ_\Pi$ is well founded.
  \end{lemma}
  \begin{proof}
    Assume that there is an infinite chain
    \[
      \vec{\phi}^{(1)} \succ_\Pi \vec{\phi}^{(2)} \succ_\Pi \ldots
    \]
    Since $\vec{\phi}^{(i)} \succ_\Pi \vec{\phi}^{(i+1)}$ implies $|\vec{\phi}^{(i)}|
    \geq |\vec{\phi}^{(i+1)}|$ and $>$ is well founded on $\NN$, there is an $i_0$ such that $|\vec{\phi}^{(i)}| = |\vec{\phi}^{(i+1)}|$ for all $i \geq i_0$.
    Thus, for each $i \geq i_0$, there is a unique $j_i$ such that $\phi^{(i)}_{j_i} \sqsubset \phi^{(i+1)}_{j_i}$ by the definition of $\succ_\Pi$.
    Hence, there is at least one $j$ such that $j = j_{i}$ for infinitely many values $i_1,i_2,\ldots$ of $i$.
    Therefore, we have
    \[
      \phi^{(i_1)}_{j} \sqsubset \phi^{(i_2)}_{j} \sqsubset \ldots.
    \]
    Since $\sqsubset$ is transitive and $\phi^{(i_k)}_j \sqsubset \phi^{(i_{k'})}_j$ implies $\phi^{(i_k)}_j \neq \phi^{(i_{k'})}_j$, $\{\phi^{(i_k)}_{j} \mid k \in \NN_{>0}\}$ is infinite.
    However, $\{\phi^{(i_k)}_{j} \mid k \in \NN_{>0}\}$ is a subset of $\mbp(\Pi)$, which is finite, so we derived a contradiction. \qed
  \end{proof}
  Now we are ready to prove soundness of our approach.
  \soundness*
  \begin{proof}
    For $\unsat$, assume $\Phi \leadsto^* (\Pi,\vec{\phi},\vec{B}) \leadsto \unsat$.
    Then $\vec{\phi} \in \mbp(\Pi)^*$ is a refutation by definition of \ref{calc:refute}.
    Thus, by soundness of resolution, $\mbp(\Pi)$ and hence also $\Pi$ is unsatisfiable.
    By \Cref{lem:correct}, we have $\Phi \equiv_\AA \Pi$, and thus $\Phi$ is unsatisfiable, too.

    For $\sat$, we use proof by contradiction.
    Assume that $\Phi$ is unsatisfiable and
    \begin{equation}
      \label{eq:assumption}
      \Phi \leadsto (\Pi^{(1)},\vec{\phi}^{(1)},\vec{B}^{(1)}) \leadsto \ldots \leadsto (\Pi^{(m)},\vec{\phi}^{(m)},\vec{B}^{(m)}) \leadsto \sat.
    \end{equation}
    By \Cref{lem:correct}, we have $\Phi \equiv_\AA \Pi^{(m)}$, so $\Pi^{(m)}$ is also unsatisfiable.
    Thus, $\mbp(\Pi^{(m)})$ is unsatisfiable as well.
    By completeness of binary input resolution for Horn clauses, $\mbp(\Pi^{(m)})^*$ contains at least one refutation.
    Thus, due to \Cref{lem:wf}, $\mbp(\Pi^{(m)})^*$ contains at least one $\succ_{\Pi^{(m)}}$-minimal refutation $\vec{\phi}$.

    Let $k_i = |\vec{\phi}^{(i)}|$ and $s_i = (\Pi^{(i)},\vec{\phi}^{(i)},\vec{B}^{(i)})$ for all $1 \leq i \leq m$.
    We say that $\vec{\phi}$ is \emph{disabled} in $s_i$ if there is a $0 \leq \ell \leq k_i$ and
    a $\phi \in B^{(i)}_\ell$ such that $\phi^{(i)}_j \equiv_\AA \phi_j$ for all $1 \leq j \leq \ell$ and $\phi \equiv_\AA \phi_{\ell+1}$.
    Then $\vec{\phi}$ is disabled in $s_m$.
    To see this, note that we have
    $\phi_1 \sqsubseteq B^{(m)}_0$ by definition of \ref{calc:accept}, and $\phi_1 \not\sqsubset B^{(m)}_0$ as $\vec{\phi}$ is $\succ_{\Pi^{(m)}}$-minimal.
    Hence there is a $\phi \in B^{(m)}_0$ such that $\phi_1 \equiv_\AA \phi$, so
    $\vec{\phi}$ is disabled in $s_m$.
    Let $1 \leq i < m$ be the largest index such that
    $\vec{\phi}$ is enabled in $s_i$.
    Note that such an index exists, as $\vec{\phi}$ is enabled in $s_1$, since $\vec{B}^{(1)} = [\emptyset]$ by definition of \ref{calc:init}.

    If $s_i \overset{\textsc{S}}{\leadsto} s_{i+1}$, then $\vec{\phi}$ is disabled in $s_i$ iff it is disabled in $s_{i+1}$, as
    \[
      B^{(i+1)}_{k_{i+1}} = \emptyset \quad \text{and} \quad  \forall j \in \{1,\ldots,k_{i+1}-1\}.\ \phi^{(i+1)}_j = \phi^{(i)}_j \land B^{(i+1)}_j = B^{(i)}_j
    \]
    by definition of \ref{calc:step}.
    Hence, $s_i \not\overset{\textsc{S}}{\leadsto} s_{i+1}$.

    If $s_i \overset{\textsc{A}}{\leadsto} s_{i+1}$, then
    \[
      B^{(i+1)}_{k_{i+1}} = \{\phi_{k_{i+1}}^{(i+1)}\} \quad \text{and} \quad \forall j \in \{1,\ldots,k_{i+1}-1\}.\ \phi^{(i+1)}_j = \phi^{(i)}_j \land B^{(i+1)}_j = B^{(i)}_j
    \]
    by definition of \ref{calc:accelerate}.
    Hence, since $\vec{\phi}$ is enabled in $s_i$, but disabled in $s_{i+1}$, we get
    $\phi_j^{(i+1)} \equiv_\AA \phi_j$ for all $1 \leq j \leq k_{i+1}$ and $\phi_{k_{i+1}}^{(i+1)} \equiv_\AA \phi_{k_{i+1}+1}$.
    So in particular, we have $\phi_{k_{i+1}} \equiv_\AA \phi_{k_{i+1}+1} \equiv_\AA \phi_{k_{i+1}}^{(i+1)}$.
    By definition of \ref{calc:accelerate}, there is a $\vec{\phi}^\circlearrowleft \in \mbp(\Pi^{(i)})^*$ such that $\phi_{k_{i+1}}^{(i+1)} = \accel(\vec{\phi}^\circlearrowleft)$.
    Hence, we have $[\phi_{k_{i+1}}^{(i+1)},\phi_{k_{i+1}}^{(i+1)}] \sqsubseteq \phi_{k_{i+1}}^{(i+1)}$ and thus $[\phi_{k_{i+1}}, \phi_{k_{i+1}+1}] \sqsubseteq \phi_{k_{i+1}}^{(i+1)}$ by \Cref{def:accel}, contradicting $\succ_{\Pi^{(m)}}$-minimality of $\vec{\phi}$.

    If $s_i \overset{\textsc{C}}{\leadsto} s_{i+1}$ or $s_i \overset{\textsc{B}}{\leadsto} s_{i+1}$, then we have
    \begin{align*}
      B^{(i+1)}_{k_{i+1}} & {} = B^{(i)}_{k_{i+1}} \cup \{\phi_{k_{i}}^{(i)}\},\\
      \forall j \in \{1,\ldots,k_{i+1}\}.\ \phi^{(i+1)}_j & {} = \phi^{(i)}_j, & \text{and} \\
      \forall j \in \{1,\ldots,k_{i+1}-1\}.\ B^{(i+1)}_j & {} = B^{(i)}_j
    \end{align*}
    by definition of $\bt$.
    Since $\vec{\phi}$ is enabled in $s_i$, but disabled in $s_{i+1}$, we get
    $\phi_j^{(i+1)} = \phi_j^{(i)} \equiv_\AA \phi_j$ for all $1 \leq j \leq k_{i+1} = k_i-1$ and $\phi_{k_{i}}^{(i)} \equiv_\AA \phi_{k_{i}}$.
    So in particular, we have $\phi_j^{(i)} \equiv_\AA \phi_j$ for all $1 \leq j \leq k_i$.
    If $s_i \overset{\textsc{C}}{\leadsto} s_{i+1}$, then the definition of
    \ref{calc:covered} implies that $\vec{\phi}^{(i)}$ is not $\succ_{\Pi^{(i)}}$-minimal,
    contradicting $\succ_{\Pi^{(i)}}$-minimality
and therefore also $\succ_{\Pi^{(m)}}$-minimality
    of $\vec{\phi}$.

    If $s_i \overset{\textsc{B}}{\leadsto} s_{i+1}$, then \emph{all} rules and queries are inactive in the state $s_i$ by definition of \ref{calc:backtrack}.
    Thus, we have $\phi_{k_i+1} \sqsubseteq B^{(i)}_{k_i}$.
    Moreover, as $\vec{\phi}$ is enabled in $s_i$, we have $\phi_{k_i+1} \not\equiv_\AA \phi$ for all $\phi \in B^{(i)}_{k_i}$.
    Thus, we get $\phi_{k_i+1} \sqsubset B^{(i)}_{k_i}$, contradicting $\succ_{\Pi^{(m)}}$-minimality of $\vec{\phi}$. \qed
  \end{proof}
}

Another important property of our calculus is that it cannot get ``stuck'' in states other than $\sat$ or $\unsat$.

\begin{restatable}[Normal Forms]{theorem}{normalforms}
  \label{thm:normal forms}
  If $\Phi \leadsto^+ s$ where $s$ is in normal form w.r.t.\ $\leadsto$, then $s \in \{\sat, \unsat\}$.
\end{restatable}
\makeproof*{thm:normal forms}{
  \normalforms*
  \begin{proof}
    Let $s = (\Pi,\vec{\phi},\vec{B})$.
    If $\vec{\phi}$ is empty, then
    \ref{calc:step} or \ref{calc:accept} is applicable.
    If $\vec{\phi}$ is non-empty, then either \ref{calc:step} or \ref{calc:backtrack} is applicable.
    Hence, $s$ is not in normal form. \qed
  \end{proof}
}

Clearly, our calculus admits many unintended
derivations, e.g., by applying \ref{calc:step} over and over again with recursive CHCs instead of accelerating them.
To prevent such derivations, a \emph{reasonable strategy} is required.

\begin{definition}[Reasonable Strategy]
  \label{def:reasonable}
  We call a strategy for $\leadsto$ \emph{reasonable} if the following holds:
  \begin{enumerate}[label=(\arabic*)]
  \item \label{it:reasonable1} If $(\Pi,\vec{\phi},\vec{B}) \leadsto^+ (\Pi,\vec{\phi}::\vec{\phi}',\vec{B}')$ for some $\vec{\phi}'$ as in the definition of \ref{calc:covered}, then \ref{calc:covered} is used.
  \item \label{it:reasonable2} \ref{calc:accelerate} is used with higher preference than \ref{calc:step}.
  \item \label{it:reasonable3} \ref{calc:accelerate} is only applied to the shortest
    recursive suffix $\vec{\phi}^\circlearrowleft$ such that
    $\accel(\vec{\phi}^\circlearrowleft)$ is not redundant w.r.t.\ $\mbp(\Pi)$.
  \item \label{it:reasonable4} If $\vec{\phi} = []$, then \ref{calc:step} is only applied with facts or conditional empty clauses.
  \end{enumerate}
  We write $\rs$ for the relation that results from $\leadsto$ by imposing a reasonable strategy.
\end{definition}
\Cref{def:reasonable} \ref{it:reasonable1} ensures that we backtrack if we added a redundant sequence $\vec{\phi}'$ of CHCs to the trace.
However, for refutational completeness (\Cref{thm:refutational-complete}), it is important
that
the application of \ref{calc:covered} is only enforced if
no new clauses have been learned while constructing $\vec{\phi}'$ (i.e., $\Pi$ remains
unchanged in the derivation $(\Pi,\vec{\phi},\vec{B}) \leadsto^+ (\Pi,\vec{\phi}::\vec{\phi}',\vec{B}')$).
The reason is that after applying \ref{calc:accelerate}, the trace might have the form $\vec{\phi} = \vec{\phi}_1::\vec{\phi}_2::\accel(\phi^\circlearrowleft)$ where $\vec{\phi}_2::\accel(\phi^\circlearrowleft) \sqsubset \accel(\phi^\circlearrowleft)$ even if $\vec{\phi}_2::\phi^\circlearrowleft$ was non-redundant before learning $\accel(\phi^\circlearrowleft)$.
If we enforced backtracking via \ref{calc:covered} in such situations (which would yield the
trace $\vec{\phi}_1::\vec{\phi}_2$),
 then to maintain refutational completeness, we would have to ensure that we eventually
 reach a state with the trace $\vec{\phi}_1::\accel(\phi^\circlearrowleft) \sqsubseteq \vec{\phi}$.
 However, this cannot be guaranteed,
 since our calculus does not terminate in general (see \Cref{thm:non-termination}).

\Cref{def:reasonable} \ref{it:reasonable2} ensures that we do not ``unroll'' recursive derivations more than once via \ref{calc:step}, but learn new clauses that cover arbitrarily many unrollings via \ref{calc:accelerate} instead.

\Cref{def:reasonable} \ref{it:reasonable3} has two purposes:
First, it prevents us from learning redundant clauses, as we must not apply \ref{calc:accelerate} if $\accel(\vec{\phi}^\circlearrowleft)$ is redundant.
Second, it ensures that we accelerate ``short'' recursive suffixes first.
The reason is that if $\vec{\phi} = \vec{\phi}_1 :: \vec{\phi}_2 :: \vec{\phi}_3$ where $\vec{\phi}_2 :: \vec{\phi}_3$ and $\vec{\phi}_3$ are recursive, then
\begin{alignat*}{2}
  && \textstyle \ground(\accel(\vec{\phi}_2 :: \vec{\phi}_3))  & \textstyle {} \overset{\rm \Cref{def:accel}}{=} \bigcup_{n \in \NN_{\geq 1}} \ground((\vec{\phi}_2 :: \vec{\phi}_3)^n) \\
  & \textstyle {} \subseteq {} & \textstyle \bigcup_{n \in \NN_{\geq 1}} \bigcup_{m \in \NN_{\geq 1}}
  \ground((\vec{\phi}_2::\vec{\phi}_3^m)^n) & \textstyle {} \overset{\rm \Cref{def:accel}}{=} \ground(\accel(\vec{\phi}_2 :: \accel(\vec{\phi}_3))),
\end{alignat*}
but the other direction (``$\supseteq$'') does not hold in general.
So in this way, we learn more general clauses.

\Cref{def:reasonable}  \ref{it:reasonable4} ensures that the first element of $\vec{\phi}$ is always a fact or a conditional empty clause.
For unsatisfiable CHC problems, the reason is that \ref{calc:refute} will never apply if $\vec{\phi}$ starts with a rule or a query.
For satisfiable CHC problems, \ref{calc:accept} only applies if all facts and conditional empty clauses are blocked.
But in order to block them eventually, we have to add them to the trace via \ref{calc:step}, which is only possible if $\vec{\phi}$ is empty.

Despite the restrictions in \Cref{def:reasonable}, our calculus is still
refutationally complete.

\begin{restatable}[Refutational Completeness]{theorem}{completeness}
  \label{thm:refutational-complete}
  If $\Phi$ is unsatisfiable, then
  \[
    \Phi \rs^* \unsat.
  \]
\end{restatable}
\begin{proofsketch}
  Given a refutation $\vec{\phi}$, one can inductively define a derivation $\Phi \rs^* \unsat$ where each step applies \ref{calc:accelerate} or \ref{calc:step}.
  For the latter, it is crucial to choose the next clause in such a way that it
  corresponds to as many steps from $\vec{\phi}$ as possible, and that it is maximal
  w.r.t.\ $\sqsubset$, to avoid the necessity to backtrack via \ref{calc:covered}. \qed
\end{proofsketch}
\makeproof*{thm:refutational-complete}{
  \completeness*
  \begin{proof}
    Due to completeness of binary input resolution for Horn Clauses, there is a refutation $\vec{\phi} \in \mbp(\Phi)^*$.
    We inductively define a $\rs$-derivation of the form
    \begin{align*}
      \Phi & {} \rs (\Phi,[],[\emptyset]) =  (\Pi_{0},\vec{\phi}_{0},\vec{B}_{0}) = s_{0} \\
           & {} \rs (\Phi,[\phi_1],[\emptyset,\emptyset]) = (\Pi_{1},\vec{\phi}_{1},\vec{B}_{1}) = s_{1} \\
           & {} \rs (\Pi_{2},\vec{\phi}_{2},\vec{B}_{2}) = s_{2} \\
           & {} \rs \ldots \\
           & {} \rs (\Pi_{k},\vec{\phi}_{k},\vec{B}_{k}) = s_{k} \\
           & {} \rs \unsat
    \end{align*}
    such that $J_i = \{ j \mid 1 \leq j \leq |\vec{\phi}|, [\phi_m]_{m=1}^{j}
    \sqsubseteq \vec{\phi}_{i}\}$ is non-empty for each $1 \leq i \leq k$.
    Here, $\phi_m$ are the elements of $\vec{\phi}$.
    Then the claim follows immediately. In the following, $j_i$ denotes $\max(J_i)$.
    Given the state $s_{i} = (\Pi_{i},\vec{\phi}_{i},\vec{B}_{i})$, we distinguish the following cases:
    \begin{enumerate}[label=(\arabic*)]
    \item \label{it:reasonable-proof1}
      If $\vec{\phi}_{i}$ has a minimal recursive suffix $\vec{\phi}^\circlearrowleft$
      such that $\accel(\vec{\phi}^\circlearrowleft)$ is not redundant, then we apply
      \ref{calc:accelerate}, resulting in the state $s_{i+1}$.
    \item \label{it:reasonable-proof2}
      Otherwise, let $c \in \NN_{\geq 1}$ be the maximal number such that $[\phi_m]_{m=j_i+1}^{j_i+c} \sqsubseteq \mbp(\Pi_{i})$.
      Then we apply \ref{calc:step} with some
      \[
        \pi \in \max_\sqsubset\left\{\pi \in \mbp(\Pi_{i}) \relmiddle{|} [\phi_m]_{m=j_i+1}^{j_i+c} \sqsubseteq \pi\right\}
      \]
      (which may not be unique, as $\sqsubset$ is a partial order), resulting in the state $s_{i+1}$ with  $j_{i+1} = j_i + c$.
    \end{enumerate}
    Note that a clause $\pi$ as in \ref{it:reasonable-proof2} always exists.
    The reason is that $[\phi_m]_{m=j_i+1}^{j_i+1} \sqsubseteq \phi_{j_i+1}$ and $\phi_{j_i+1} \in \mbp(\Pi_{i})$. Hence, there always exists a $c  \in \NN_{\geq 1}$ with $[\phi_m]_{m=j_i+1}^{j_i+c} \sqsubseteq \mbp(\Pi_{i})$.

    To show $\Phi \leadsto s_1 \leadsto \ldots \leadsto s_k$, it remains to show that the clause $\pi$  in \ref{it:reasonable-proof2} is not blocked. To
    this end, we prove
    \[
      \text{If $\pi$ is blocked in state $s_i$, then $[\phi_m]_{m=j_i+1}^{j_i+c} \not\sqsubseteq \pi$ for all $c \geq 1$.}
    \]
    This implies that  the clause $\pi$ in \ref{it:reasonable-proof2} is not blocked and thus \ref{calc:step} is applicable.
    Assume that $\pi$ is blocked in state $s_{i}$.
    Then $\pi$ is the last element of the trace by definition of \ref{calc:accelerate}, and we have $\resolve(\pi,\pi) \sqsubseteq \pi$, since this holds for every accelerated clause $\pi$.
Moreover, $j_i = j_{i-1} + c'$, where $c' \geq 0$ is the maximal number such that $[\phi_m]_{m=1}^{j_{i-1}+c'} \sqsubseteq \vec{\phi}_i$.
    Hence, there is a suffix $\vec{\pi}$ of $[\phi_m]_{m=1}^{j_{i}}$ such that $|\vec{\pi}| \geq c'$, $\vec{\pi} \sqsubseteq \pi$, and
    \begin{equation}
      \label{eq:non-redundant}
      \vec{\pi}::[\phi_m]_{m=j_i+1}^{j_i+c''} \not \sqsubseteq \pi \text{ for all } c'' > 0
    \end{equation}
    by maximality of $c'$.
    Assume that there is a $c \in \NN$ with $[\phi_m]_{m=j_i+1}^{j_i+c} \sqsubseteq \pi$.
    Then we have
    \[
      \vec{\pi}::[\phi_m]_{m=j_i+1}^{j_i+c} \sqsubseteq \resolve(\pi,\pi) \sqsubseteq \pi
    \]
    and hence \eqref{eq:non-redundant} implies $c = 0$.

    Hence, we have $\Phi \leadsto s_0 \leadsto \ldots \leadsto s_k$.
    Next, we prove that this sequence is reasonable.
    We handle each item of \Cref{def:reasonable} individually.

    \begin{description}
    \item[\Cref{def:reasonable} \ref{it:reasonable1}:] This case never applies.
      To show this, we prove that $i' < i$, $\Pi_{i'} = \Pi_i$, and $\vec{\phi}_i = \vec{\phi}_{i'}::\vec{\pi}'$ implies $\vec{\pi}' \not \sqsubset \mbp(\Pi_i)$ if $|\vec{\pi}'| = 1$, and $\vec{\pi}' \not \sqsubseteq \mbp(\Pi_i)$ if $|\vec{\pi}'| > 1$.
      Assume $\Pi_{i'} = \Pi_i$ and $\vec{\phi}_i = \vec{\phi}_{i'}::\vec{\pi}'$ for some $i' < i$.
      Then $\Pi_{i'} = \Pi_i$ implies that $s_{i'+1},...,s_i$ were constructed via \ref{it:reasonable-proof2}.
      Thus,
if $\pi_1'$ is the first CHC of the sequence $\vec{\pi}'$, then we
have $[\phi_m]_{m=j_{i'}+1}^{j_{i'}+c} \sqsubseteq \pi_1'$ for some $c \in \NN_{\geq 1}$, which has to be maximal according to \ref{it:reasonable-proof2}.
If $|\vec{\pi}'| = 1$, then $\vec{\pi}' = \pi_1' \sqsubset \mbp(\Pi_i)$
would contradict the requirement in \ref{it:reasonable-proof2} that
$\pi_1'$ is maximal w.r.t.\ $\sqsubset$.
      Assume $|\vec{\pi}'| > 1$.
      As $s_{i'+1}$ was also constructed via \ref{it:reasonable-proof2}, $|\vec{\pi}'| > 1$ implies $j_i > j_{i'+1}$.
      Hence, we get $[\phi_m]_{m=j_{i'}+1}^{j_{i}} \not\sqsubseteq \mbp(\Pi_i)$, as $j_{i} >
      j_{i'+1} = j_{i'} + c$, but $c$ is the maximal natural number such that
      $[\phi_m]_{m=j_{i'}+1}^{j_{i'}+c} \sqsubseteq \Pi_i$.
Since $[\phi_m]_{m=j_{i'}+1}^{j_{i}} \sqsubseteq \vec{\pi}'$, this implies $\vec{\pi}'
\not\sqsubseteq \mbp(\Pi_i)$.
    \item[\Cref{def:reasonable} \ref{it:reasonable2}:] Trivial, as
      \ref{it:reasonable-proof1} is applied with higher preference than
      \ref{it:reasonable-proof2}.
    \item[\Cref{def:reasonable} \ref{it:reasonable3}:] Trivial due to \ref{it:reasonable-proof1}.
    \item[\Cref{def:reasonable} \ref{it:reasonable4}:] The first \ref{calc:step} in the
      constructed $\leadsto$-sequence uses a CHC $\phi$ such that we have
        $\vec{\phi}' \sqsubseteq \phi$ for a non-empty prefix $\vec{\phi}'$ of $\vec{\phi}$.
      Since $\vec{\phi}$ is a refutation, its first element is a fact or a conditional empty clause.
      Hence, $\resolve(\vec{\phi}')$ is a fact or a conditional empty clause, too.
      Thus, $\phi$ must be a fact or conditional empty clause as well.
    \end{description}

    Next, we prove that each $J_i$ is non-empty, i.e., that for each $1 \leq
    i \leq k$, there is a $1 \leq j \leq |\vec{\phi}|$ such that $[\phi_m]_{m=1}^j \sqsubseteq \vec{\phi}_i$.
    We use induction on $1 \leq
    i \leq k$, where the case $i=1$ is immediate due to the choice of $\pi$ in \ref{it:reasonable-proof2}.
    Let $i > 1$.
    First assume that \ref{it:reasonable-proof1} was applied to the state $s_{i-1}$ and $\vec{\phi}_{i-1} = \vec{\phi}'_{i'}::\vec{\phi}^\circlearrowleft$ for some $i' \in \NN_{\geq 1}$ where $i' < i-1$.
    Then by the induction hypothesis, $J_{i-1}$ is non-empty and thus $j_{i-1}$ exists.
    We get:
    \begin{align*}
      & \ground([\phi_m]_{m=1}^{j_{i-1}}) \\
      {} \subseteq {} & \ground(\vec{\phi}_{i-1}) \tag{IH} \\
      {} = {} & \ground(\vec{\phi}_{i'}::\vec{\phi}^\circlearrowleft) \tag{as $\vec{\phi}_{i-1} = \vec{\phi}'_{i'}::\vec{\phi}^\circlearrowleft$} \\
      {} \subseteq {} & \ground(\vec{\phi}_{i'}::\accel(\vec{\phi}^\circlearrowleft)) \tag{\Cref{lem:ground}, as $\vec{\phi}^\circlearrowleft \sqsubseteq \accel(\vec{\phi}^\circlearrowleft)$} \\
      {} = {} & \ground(\vec{\phi}_{i})
    \end{align*}
    Hence, $j_{i-1} \in J_i$, i.e., $J_i$ is non-empty.
    If \ref{it:reasonable-proof2} was applied to the state $s_{i-1}$, then $J_i$ is trivially non-empty due to the choice of $\pi$ in \ref{it:reasonable-proof2}.

    Finally, note that the resulting sequence $\Phi \leadsto s_0 \leadsto \ldots \leadsto
    s_k$
    is finite: For \ref{it:reasonable-proof2}, we have $j_i < j_{i+1}$, and the sequence $j_1,j_2,\ldots$ is bounded by $|\vec{\phi}|$.
    For \ref{it:reasonable-proof1}, we have $j_i \leq j_{i+1}$, but the (finite) number of recursive infixes of $\vec{\phi}$ that are not yet redundant is decremented.
    Hence, we get a lexicographic termination argument, which ensures that there is an
    $s_k$ with $j_k = |\vec{\phi}|$.  This implies $s_k \leadsto \unsat$. \qed
  \end{proof}
}

However, in general our calculus does not terminate, even with a reasonable
strategy.
Note that even though CHC-SAT is undecidable for, e.g., CHCs over the theory LIA, non-termination of $\rs$ is not implied by soundness of ADCL.
The reason is that we assume oracles for undecidable sub-problems like SMT, checking redundancy, and acceleration.
As acceleration may introduce non-linear integer arithmetic, both SMT and checking redundancy may even become undecidable when analyzing CHCs over a decidable theory like LIA.

To prove non-termination, we extend our calculus by one additional component:
A mapping $\LL: \mbp(\Pi) \to \PP(\mbp(\Phi)^*)$ from $\mbp(\Pi)$ to regular languages
over $\mbp(\Phi)$, where $\PP(\mbp(\Phi)^*)$ denotes the power set of $\mbp(\Phi)^*$.
We will show that this mapping gives rise to an alternative characterization of the ground
instances of $\mbp(\Pi)$, which will be exploited in our non-termination proof
(\Cref{thm:non-termination}). Moreover, this mapping is
  also  used
  in our implementation to check redundancy, see \Cref{sec:implementation}.
To extend our calculus, we lift $\LL$ from $\mbp(\Pi)$ to $\mbp(\Pi)^*$ as follows:
\[
  \LL(\epsilon) \Def \epsilon \qquad \qquad \qquad \LL(\vec{\pi}::\pi) \Def \LL(\vec{\pi})::\LL(\pi)
\]
Here, ``$::$'' is also used to denote language concatenation, i.e., we have
\[
  \LL_1 :: \LL_2 \Def \{ \vec{\pi}_1 :: \vec{\pi}_2 \mid \vec{\pi}_1 \in \LL_1, \vec{\pi}_2 \in \LL_2 \}.
\]
So while we lift other notations to sequences of transitions via
resolution, $\LL(\vec{\tau})$ does \emph{not} stand for $\LL(\resolve(\vec{\tau}))$.

\def\scale{0.875}
\def\width{0.89\textwidth}
\begin{definition}[ADCL with Regular Languages]
  \label{def:calc-reg}
  We extend states (see \Cref{def:state}) by a fourth component $\LL: \mbp(\Pi) \to \PP(\mbp(\Phi)^*)$.
  The rules \ref{calc:init} and \ref{calc:accelerate} of the ADCL calculus (see
  \Cref{def:calc}) are adapted as follows:
  \begin{align}
    & \scalebox{\scale}{\parbox{\width}{
      \[
        \infer{
        \Phi \rs (\Phi,[],[\emptyset],\LL)
        }{\LL(\phi) = \{\phi\} \text{ for all $\phi \in \mbp(\Phi)$}}
      \]
    }}
    \tag{\textsc{Init}} \label{calc:init-reg} \\
    & \scalebox{\scale}{\parbox{\width}{
      \[ \hspace*{-.2cm}
        \infer{
          (\Pi,\vec{\phi}::\vec{\phi}^\circlearrowleft,\vec{B}::\vec{B}^\circlearrowleft,\LL) \rs (\Pi \cup \{\phi\},\vec{\phi}::\phi,\vec{B}::\{\phi\},\LL')
        }{
          \vec{\phi}^\circlearrowleft \text{ is recursive} & |\vec{\phi}^\circlearrowleft|
          = |\vec{B}^\circlearrowleft| & \accel(\vec{\phi}^\circlearrowleft) = \phi & \LL'
          = \LL \uplus  (\phi \mapsto \LL(\vec{\phi}^\circlearrowleft)^+)
        }
      \]
    }}
    \tag{\textsc{Accelerate}} \label{calc:accelerate-reg}
  \end{align}
  All other rules from \Cref{def:calc} leave the last component of the state unchanged.
\end{definition}
Here, $\LL(\pi)^+$ denotes the ``Kleene plus'' of $\LL(\pi)$, i.e., we have
\[
  \LL(\pi)^+ \Def \bigcup_{n \in \NN_{\geq 1}} \LL(\pi)^n.
\]
Note that \Cref{def:calc-reg} assumes a reasonable strategy (indicated by the notation $\rs$).
Hence, when \ref{calc:accelerate} is applied, we may assume $\phi \notin \mbp(\Pi) = \dom(\LL)$.
Otherwise, $\phi$ would be redundant and hence a reasonable strategy would not allow the
application of \ref{calc:accelerate}.
For this reason, we may write ``$\uplus$'' in the definition of $\LL'$.

The following lemma allows us to characterize the ground instances of elements of $\mbp(\Pi)$ via $\LL$.
Here, we lift $\ground$ to sets by defining $\ground(X) \Def \bigcup_{x \in X} \ground(x)$, where $X$ may be a set of CHCs or a language over CHCs.
Thus, $\ground(\LL(\pi))$ is the set of all ground instances of the final resolvents of the sequences in $\LL(\pi)$.

\begin{restatable}{lemma}{redundant}
  \label{lem:redundant}
  If  $\Phi \rs^* (\Pi,\vec{\phi},\vec{B},\LL)$ and $\pi \in \mbp(\Pi)$, then $\ground(\pi) = \ground(\LL(\pi))$.
\end{restatable}
\makeproof*{lem:redundant}{
  \redundant*
  \begin{proof}
    We use induction on the length of the $\rs$-derivation.
    If $\Phi \rs (\Pi,\vec{\phi},\vec{B},\LL)$, then the claim follows from the definition of \ref{calc:init-reg}.
    Assume
    \[
      \Phi \rs^* (\Pi',\vec{\phi}',\vec{B}',\LL') \rs (\Pi,\vec{\phi},\vec{B},\LL).
    \]
    If the last step is not \ref{calc:accelerate}, then the claim follows from the induction hypothesis.
    Assume that the last step is \ref{calc:accelerate}.
    Then we have $\Pi = \Pi' \cup \{\phi\}$ and $\LL = \LL' \uplus \phi \mapsto \LL'(\vec{\phi}^\circlearrowleft)^+$.
    If $\pi \neq \phi$, then the claim follows from the induction hypothesis, so assume $\pi = \phi$.
    Then it remains to show
    \begin{equation}
      \label{main statement lemma redundant}
      \ground(\phi) = \ground(\LL(\phi)) = \ground(\LL'(\vec{\phi}^\circlearrowleft)^+),
      \text{ where } \phi = \accel(\vec{\phi}^\circlearrowleft).
    \end{equation}
    By the induction hypothesis, we have $\ground(\phi') = \ground(\LL'(\phi'))$ for all $\phi' \in \vec{\phi}^\circlearrowleft$.
    We first show
    \begin{equation}
      \label{eq:resolve-reg}
      \ground(\vec{\phi}^\circlearrowleft) = \ground(\LL'(\vec{\phi}^\circlearrowleft)).
    \end{equation}
    To this end, we show that $\vec{\pi} \in \mbp(\Pi')^*,\pi \in \mbp(\Pi')$, $\ground(\vec{\pi}) = \ground(\LL'(\vec{\pi}))$ and $\ground(\pi) = \ground(\LL'(\pi))$ imply $\ground(\vec{\pi}::\pi) = \ground(\LL'(\vec{\pi}::\pi))$.
    Then \eqref{eq:resolve-reg} follows by a straightforward induction over $|\vec{\phi}^\circlearrowleft|$.
    We have:
    \begin{align*}
      & \ground(\vec{\pi}::\pi) \\
      {} = {} & \{\phi \in \resolve(\ground(\vec{\pi}),\ground(\pi)) \mid {} \models_\AA \cond(\phi) \} \tag{\Cref{{lem:distribute}}} \\
      {} = {} & \{\phi \in \resolve(\ground(\LL'(\vec{\pi})),\ground(\LL'(\pi))) \mid {} \models_\AA \cond(\phi) \} \tag{IH} \\
      {} = {} & \ground(\resolve(\LL'(\vec{\pi}),\LL'(\pi))) \tag{\Cref{{lem:distribute}}} \\
      {} = {} & \ground(\resolve(\LL'(\vec{\pi}::\pi))) \tag{\Cref{lem:distribute2}} \\
      {} = {} & \ground(\LL'(\vec{\pi}::\pi)) \tag{by the definition of $\ground$}
    \end{align*}
    Now we prove \eqref{main statement lemma redundant}:
    \begin{align*}
      & \ground(\accel(\vec{\phi}^\circlearrowleft)) \\
      {} = {} & \bigcup_{n \in \NN_{\geq 1}} \ground((\vec{\phi}^\circlearrowleft)^n) \tag{\Cref{def:accel}} \\
      {} = {} & \bigcup_{n \in \NN_{\geq 1}} \{ \phi \in \ground(\vec{\phi}^\circlearrowleft)^n \mid {} \models_\AA \cond(\phi) \} \tag{\Cref{{lem:distribute}}} \\
      {} = {} & \bigcup_{n \in \NN_{\geq 1}} \{ \phi \in \ground(\LL'(\vec{\phi}^\circlearrowleft))^n \mid
      {} \models_\AA \cond(\phi) \} \tag{by \eqref{eq:resolve-reg}} \\
      {} = {} & \bigcup_{n \in \NN_{\geq 1}} \ground(\LL'(\vec{\phi}^\circlearrowleft)^n) \tag{\Cref{{lem:distribute}}} \\
      {} = {} & \ground(\LL'(\vec{\phi}^\circlearrowleft)^+) \tag{by the definition of Kleene plus}
    \end{align*}
    \qed
  \end{proof}
}

Now we are ready to prove that, even with a reasonable strategy, ADCL does not
terminate.

\begin{restatable}[Non-Termination]{theorem}{nonterm}
  \label{thm:non-termination}
  There exists a satisfiable CHC problem $\Phi$ such that $\Phi \not\rs^* \sat$.
  Thus, $\rs$ does not terminate.
\end{restatable}
\begin{proofsketch}
   One can construct a satisfiable CHC problem $\Phi$
   such that all (infinitely many) resolution sequences with $\Phi$ are \emph{square-free}, i.e., they do not contain a non-empty subsequence of the form $\vec{\phi}::\vec{\phi}$.
   For example,
   this can be achieved by encoding the differences between subsequent numbers of the Thue-Morse sequence
  \cite{oeis-thue-morse,oeis-thue-morse-difference}.
  As an invariant of our calculus, $\LL(\Pi)$ just contains finitely many square-free
  words for any reachable state $(\Pi,\vec{\phi},\vec{B},\LL)$.
  As $\ground(\Pi) = \ground(\LL(\Pi))$, this means that $\Pi$ cannot cover
  all resolution sequences with $\Phi$.
  Thus, the assumption $\Phi \rs^* \sat$ results in a contradiction.
  \qed
\end{proofsketch}
\makeproof*{thm:non-termination}{
  \nonterm*
  \begin{proof}
    ~\\[-2.5em]
    \section*{Proof Idea:}

    We will construct a satisfiable CHC problem $\Phi$ such that $\Phi \not\rs^* \sat$.
    Since $\Phi$ is satisfiable, we also have $\Phi \not\rs^* \unsat$ by the soundness of our calculus (\Cref{thm:soundness}).
    But by \Cref{thm:normal forms}, the only normal forms w.r.t.
    $\leadsto$ (and also w.r.t.\ $\rs$) are $\sat$ and $\unsat$.
    Therefore, $\rs$ does not terminate when starting with $\Phi$.
    In fact, it is not even normalizing (i.e., there does not \emph{exist} any derivation to a normal form).

    \section*{Construction of $\Phi$:}

    To construct $\Phi$, we consider the Thue-Morse sequence $[v_i]_{i \in \NN}$
    \cite{oeis-thue-morse}.
    Let $w_i \Def v_{i+1} - v_i$.
    The resulting infinite sequence $[w_i]_{i \in \NN}$ over the alphabet $\{-1,0,1\}$ is well-known to be \emph{square-free}, i.e., it does not contain a non-empty infix $u :: u$ \cite{oeis-thue-morse-difference}.
    Let $\Phi$ be a satisfiable CHC problem that contains
    \begin{equation}
      \label{eq:nt1}
      \tag{\protect{\ensuremath{\phi_f}}}
      \top \implies \pred{ThueMorse}(0,1),
    \end{equation}
    the following CHCs $\Phi_\pred{ThueMorse}$,
    \begin{align*}
      \pred{ThueMorse}(I,X) \land X = -1 \land J = I+1 & {} \implies \pred{Next}(J)  \\
      \pred{ThueMorse}(I,X) \land X = 0 \land J = I+1  & {} \implies \pred{Next}(J)  \\
      \pred{ThueMorse}(I,X) \land X = 1 \land J = I+1  & {} \implies \pred{Next}(J),
    \end{align*}
    and CHCs $\Phi_{\pred{Next}}$ that are equivalent to
    \begin{equation}
      \pred{Next}(I) \land X = w_I \implies \pred{ThueMorse}(I,X).
    \end{equation}
    Note that $\Phi_{\pred{Next}}$ exists, since $[w_i]_{i \in \NN}$ is computable and linear CHCs are Turing complete.
    W.l.o.g., we assume $\mbp(\Phi) = \Phi$ in the sequel.
    Then for each $n \in \NN$, there is a unique sequence
    \[
      \vec{\phi}_n \in \ref{eq:nt1} :: (\Phi_{\pred{ThueMorse}} :: \Phi_{\pred{Next}}^+)^{n}
    \]
    such that $\ground(\vec{\phi}_n) = \{\top \implies \pred{ThueMorse}(n,w_n)\}$.
    Hence, $\proj{\vec{\phi}_n}{\Phi_\pred{ThueMorse}}$ (i.e., the sequence that results from $\vec{\phi}_n$ by omitting all CHCs that are not contained in $\Phi_\pred{ThueMorse}$) is square-free.

    \section*{Contradicting $\Phi \rs^* \sat$:}

    Assume
    \[
      \Phi \rs s_1 \rs \ldots \rs s_k = (\Pi,[],[B],\LL) \rs \sat.
    \]

    \subsection*{1. Finitely many square-free words in $\proj{\LL(\Pi)}{{\Phi_\pred{ThueMorse}}}$:}

    We first prove that any finite union $\LL = \bigcup_{i=1}^c \LL_i$ of languages $\LL_i$ over a finite alphabet that are built from singleton languages, concatenation, and Kleene plus just contains finitely many square-free words.
    To this end, it suffices to prove that each $\LL_i$ contains at most one square-free word.
    Let $1 \leq i \leq c$ be arbitrary but fix.
    We use induction on $\LL_i$.
    If $\LL_i$ is a singleton language, then the claim is trivial.
    If $\LL_i = \LL'::\LL''$, then the square-free words in $\LL_i$ are a subset of
    \begin{equation}
      \label{eq:square-free}
      \{w'::w'' \mid w' \in \LL', w'' \in \LL'', w' \text{ and } w'' \text{ are square-free}\}.
    \end{equation}
    By the induction hypothesis, there is at most one square-free word $w' \in \LL''$ and at most one square-free word $w'' \in \LL''$.
    Thus, the size of \eqref{eq:square-free} is at most one.
    If $\LL_i = (\LL')^+$, then $\LL_i$ contains the same square-free words as $\LL'$.
    To see this, let $w \in (\LL')^+$ be square-free.
    Then there are $n \in \NN_{\geq 1}$, $v_1,\ldots,v_n \in \LL'$ such that $v_1::\ldots::v_n = w$.
    Since $w$ is square-free, each $v_i$ must be square-free, too.
    As $\LL'$ contains at most one square-free word by the induction hypothesis, we get
    $v_1 = \ldots = v_n$.
    Since $w$ is square-free, this implies $n = 1$ (or that $w$ and the $v_i$ are the
    empty word).

    Note that
    \[
      \LL(\Pi)= \bigcup_{\pi \in \Pi} \LL(\pi)
    \]
    is a finite union of languages over $\Phi$ that are built from singleton languages, concatenation, and Kleene plus.
    Then
    \[
      \proj{\LL(\Pi)}{\Phi_\pred{ThueMorse}} \Def \{\proj{w}{\Phi_\pred{ThueMorse}}\mid w \in \LL(\Pi)\}
    \]
    is a finite union of languages over $\Phi_\pred{ThueMorse}$ that are built from singleton languages, concatenation, and Kleene plus.
    Thus, $\proj{\LL(\Pi)}{\Phi_\pred{ThueMorse}}$ just contains finitely many square-free words.

    \subsection*{2. Extending $\Phi$ to an unsatisfiable CHC problem $\Phi'$:}

    Thus, there is an $m > 1$ such that $\proj{\vec{\phi}_m}{\Phi_\pred{ThueMorse}} \notin \proj{\LL(\Pi)}{\Phi_\pred{ThueMorse}}$.
    To see that, note that the words $\proj{\vec{\phi}_n}{\Phi_\pred{ThueMorse}}$ are square-free and pairwise different, but $\proj{\LL(\Pi)}{\Phi_\pred{ThueMorse}}$ only contains finitely many square-free words.
    Let
    \begin{align*}
      \phi_q & {}\Def \pred{ThueMorse}(I,X) \land I = m \land X = w_m \implies \bot, \\
      \Phi'  & {} \Def \Phi \cup \{\phi_q\}, & & \text{and} \\
      s_i    & {} \Def (\Pi_i,\vec{\pi}_i,\vec{B}_i,\LL_i) & & \text{for each $i \in [1,k]$.}
    \end{align*}

    \subsection*{3. $\Phi' \leadsto^* \sat$:}

    Then we get
    \[
      \Phi' \rs s'_1 \rs \ldots \rs s'_k \rs \sat
    \]
    where for each $1 \leq i \leq k$, we have
    \[
      \Pi'_i \Def \Pi_i \cup \{\phi_q\}, \qquad \LL'_i \Def \LL_i \uplus \phi_q \mapsto \{\phi_q\}, \qquad \text{and} \qquad s'_i \Def (\Pi'_i, \vec{\pi}_i, \vec{B}_i,\LL'_i).
    \]
    To prove this, it suffices to show $\Phi' \rs^k s'_k$, as we trivially have $s'_k \leadsto \sat$ since $s'_k$ and $s_k$ contain the same facts and no conditional empty clauses.
    To this end, we prove $\Phi' \rs^i s'_i$ for all $1 \leq i \leq k$ by induction on $i$.

    \subsubsection*{3.1. Induction Base:}

    For $i = 1$, the only applicable rule is \ref{calc:init}, i.e., we have
    \begin{align*}
      \Phi & {} \overset{\textsc{I}}{\leadsto}_{\mathit{rs}} (\Phi,[],[\emptyset],\phi \mapsto \{\phi\}) = (\Pi_1,\vec{\pi}_1,\vec{B}_1,\LL_1) & \text{and} \\
      \Phi' & {} \overset{\textsc{I}}{\leadsto}_{\mathit{rs}} (\Phi',[],[\emptyset],\phi \mapsto \{\phi\}) = (\Pi'_1,\vec{\pi}_1,\vec{B}_1,\LL'_1).
    \end{align*}

    \subsubsection*{3.2. Induction Step -- Trivial Cases:}

    For $i > 1$, we perform a case analysis on the rule that is used for the step $s_{i-1} \rs s_i$.
    \ref{calc:init} does not apply to $s_{i-1}$.
    For \ref{calc:step}, if $\phi$ was active in the $i^{th}$ step starting from $\Phi$, then it is also active in the $i^{th}$ step starting from $\Phi'$, as $\vec{B}_{i-1}$ is the same in both derivations.
    Moreover, we get the same set $B$ and thus, the same $\vec{B}_i$, since the same clauses are learned in the first $i-1$ (identical) steps of both ${\rs}$-derivations.
    For \ref{calc:covered}, we have $\vec{\phi}' \sqsubset \Pi_i$ iff $\vec{\phi}' \sqsubset \Pi'_i$, and $\vec{\phi}' \sqsubseteq \Pi_i$ iff $\vec{\phi}' \sqsubseteq \Pi'_i$.
    The reason is that every model of
$\Phi$ is also a model of 
    $\vec{\phi}'$ and thus $\vec{\phi}' \sqsubset \phi_q$ would imply unsatisfiability of $\Phi$, but $\Phi$ is satisfiable.
    For the same reason, we have $\accel(\vec{\phi}^\circlearrowleft) \sqsubseteq \Pi_i$ iff $\accel(\vec{\phi}^\circlearrowleft) \sqsubseteq \Pi'_i$ in the case of \ref{calc:accelerate}.

    \subsubsection*{3.3. Induction Step -- \ref{calc:backtrack}:}

    For \ref{calc:backtrack}, we have to show that $\phi_q$ is inactive.
    Since $\phi_q$ is clearly not blocked, we have to show that $\cond(\vec{\pi}_{i-1}::\phi_q)$ is unsatisfiable.
    Assume otherwise.
    Then we have:
    \begin{align*}
      \vec{\phi}_m          & {} = \ref{eq:nt1}::\vec{\phi}                                    &  & \text{for some } \vec{\phi} \in \Phi^*    \\
      \vec{\pi}_{i-1}       & {} = \ref{eq:nt1}::\vec{\pi}::\pi                                &  & \text{for some } \vec{\pi}::\pi \in \Pi^+ \\
      \top                  & {} \implies \pred{ThueMorse}(m,w_m) \in \ground(\vec{\pi}_{i-1})                                                \\
      \pred{ThueMorse}(0,1) & {} \implies \pred{ThueMorse}(m,w_m) \in \ground(\vec{\pi}::\pi)
    \end{align*}
    To see why the equality $\vec{\pi}_{i-1} = \ref{eq:nt1}::\vec{\pi}::\pi$ holds (i.e., why $|\vec{\pi}_{i-1}| > 1$), note that $\vec{\pi}_{i-1} = [\ref{eq:nt1}]$ would imply $m=0$, contradicting $m>1$ (i.e., then $\cond(\vec{\pi}_{i-1}::\phi_q)$ would be unsatisfiable).

    \paragraph*{\underline{Step 3.3.1.}
      We show $\vec{\phi} \in \LL(\Pi)$:}

    Note that some clause from $\Phi_{\pred{ThueMorse}}$ is blocked in state $s_{i-1}$ (as \ref{calc:backtrack}
    applies to $s_{i-1}$). The reason is that the head symbol of the last clause in $\vec{\pi}_{i-1}$ is $\pred{ThueMorse}$. Thus, one of the clauses $\phi'$ of $\Phi_{\pred{ThueMorse}}$ would be applicable (i.e., $\cond(\vec{\pi}_{i-1}::\phi')$ is satisfiable). But since the clause $\phi'$ is inactive, it must be blocked.
    Thus, there is a $j < i-1$ such that $\vec{\pi}_j = \vec{\pi}_{i-1}$ and no element of $\Phi_{\pred{ThueMorse}}$ is blocked (right after adding $\pi$ to the trace via \ref{calc:accelerate} or \ref{calc:step}). We cannot have $j = i-1$ because when adding $\pi$ to the trace via \ref{calc:accelerate} or \ref{calc:step}, we reach a state where no clause of $\Phi_{\pred{ThueMorse}}$ is blocked.

    \paragraph*{Case 3.3.1.1. $\vec{\pi}$ contains clauses with body symbol $\pred{ThueMorse}$:}

    If $\vec{\pi}$ contains $c>0$ clauses with body symbol $\pred{ThueMorse}$, then $\pi_1$ is one of them, as $\vec{\pi}_j$ starts with $\ref{eq:nt1}$.
    Thus, $\vec{\pi}::\pi$ is recursive, i.e., $\vec{\pi}_j$ has a recursive suffix with ground instance
    \begin{equation}
      \label{eq:m-steps}
      \pred{ThueMorse}(0,1) \implies \pred{ThueMorse}(m,w_m).
    \end{equation}
    If $\accel(\vec{\pi}::\pi) \sqsubseteq \Pi_j$, then $\Pi_j \subseteq \Pi$ contains a clause with ground instance \eqref{eq:m-steps}.

    Otherwise, a reasonable strategy needs to apply \ref{calc:accelerate} in state $s_j$.
    Then we get $\accel(\vec{\pi}::\pi) \sqsubseteq \Pi$ by induction on $c$:
    If $c=1$, then $\accel(\vec{\pi}::\pi) \in \Pi_{j+1} \subseteq \Pi$.
    If $c>1$ and $\accel(\vec{\pi}::\pi) \not\sqsubseteq \Pi_{j+1}$, then $\vec{\pi}_{j+1}$ has a recursive suffix $\vec{\phi}^\circlearrowleft \sqsubseteq \vec{\pi}::\pi$ (due to \Cref{lem:assoc}) such that $\accel(\vec{\phi}^\circlearrowleft)$ is not redundant, and hence a reasonable strategy again needs to apply \ref{calc:accelerate} in state $s_{j+1}$.
    Thus, $\accel(\vec{\pi}::\pi) \sqsubseteq \Pi$ follows from the induction hypothesis

    So $\Pi$ contains a clause with ground instance \eqref{eq:m-steps}.
    Thus, we have
    \begin{equation}
      \label{eq:contains-ground}
      \eqref{eq:m-steps} \in \ground(\Pi) \overset{\rm{\Cref{lem:redundant}}}{=} \ground(\LL(\Pi)).
    \end{equation}
    Note that for all $\vec{\phi}' \in \Phi^*$, we have
    \[
      \eqref{eq:m-steps} \in \ground(\vec{\phi}') \qquad \text{iff} \qquad \vec{\phi} = \vec{\phi}'
    \]
    since otherwise, $\vec{\phi}_m$ would not be unique.
    Hence, \eqref{eq:contains-ground} implies $\vec{\phi} \in \LL(\Pi)$.

    \paragraph*{Case 3.3.1.2. $\vec{\pi}$ contains no clause with body symbol $\pred{ThueMorse}$:}

    Now consider the case that $\vec{\pi}$ does not contain a clause with body symbol $\pred{ThueMorse}$.
    As $\vec{\pi}_j$ starts with \ref{eq:nt1}, this implies $\vec{\pi} = []$.
    So we get
    \[
      \eqref{eq:m-steps} \in \ground(\pi) \overset{\rm{\Cref{lem:redundant}}}{=} \ground(\LL(\pi)) \overset{\pi \in \Pi}{\subseteq} \ground(\LL(\Pi)).
    \]
    Hence, $\vec{\phi} \in \LL(\Pi)$ follows as in the previous case.

    \paragraph*{\underline{Step 3.3.2.} We show that $\vec{\phi} \in \LL(\Pi)$ implies $\proj{\vec{\phi}_m}{\Phi_\pred{ThueMorse}} \in \proj{\LL(\Pi)}{\Phi_\pred{ThueMorse}}$:}

    We have shown $\vec{\phi} \in \LL(\Pi)$. Hence, we have
    \[
      \proj{\vec{\phi}}{\Phi_\pred{ThueMorse}} \in \proj{\LL(\Pi)}{\Phi_\pred{ThueMorse}}.
    \]
    As $\proj{\vec{\phi}}{\Phi_\pred{ThueMorse}} = \proj{\vec{\phi}_m}{\Phi_\pred{ThueMorse}}$, this contradicts the observation
    \[
      \proj{\vec{\phi}_m}{\Phi_\pred{ThueMorse}} \notin \proj{\LL(\Pi)}{\Phi_\pred{ThueMorse}}.
    \]

    \section*{4. $\Phi' \leadsto^* \sat$ contradicts Soundness}

    This contradicts the soundness of $\rs$, because $\Phi' = \Phi \cup \{\phi_q\}$ is unsatisfiable.
    The reason is that $\sigma \models_\AA \Phi$ implies $\sigma \models_\AA \pred{ThueMorse}(n,w_n)$ for all $n \in \NN$.
    Hence, our assumption was wrong and we have $\Phi \not\rs^* \sat$. \qed
  \end{proof}
}
The construction from the proof of \Cref{thm:non-termination} can also be used to show that there are non-terminating derivations $\Phi \rs s_1 \rs s_2 \rs \ldots$ where $\Phi$ is unsatisfiable.
However, in this case there is also another derivation $\Phi \rs^* \unsat$ due to refutational completeness (see \Cref{thm:refutational-complete}).

%% file: implementation.tex
\section{Implementing ADCL}
\label{sec:implementation}

We now explain how we implemented ADCL efficiently in \anonymous{our}{the} tool \tool{LoAT}.
Here we focus on proving unsatisfiability.
The reason is that our implementation cannot prove $\sat$ at the moment, since it uses certain approximations that are incorrect for $\sat$, as detailed below.
Thus, when applying \ref{calc:accept}, our implementation returns $\unknown$ instead of $\sat$.
Our implementation uses \tool{Yices} \cite{yices} and \tool{Z3} \cite{z3} for SMT solving.
Moreover, it is based on the acceleration technique from \cite{acceleration-calculus}, whose implementation solves recurrence relations with \tool{PURRS} \cite{purrs}.

\subsubsection*{Checking Redundancy} To check redundancy in \ref{calc:accelerate} (as
required for reasonable strategies in \Cref{def:reasonable}), we use the fourth component $\LL$ of states introduced in \Cref{def:calc-reg}.
More precisely, for \ref{calc:accelerate}, we check if $\LL(\vec{\phi}^\circlearrowleft)^+ \subseteq \LL(\phi)$ holds for some learned clause $\phi$.
In that case, $\accel(\vec{\phi}^\circlearrowleft)$ is redundant due to \Cref{lem:redundant}.
Since $\LL(\vec{\phi}^\circlearrowleft)^+ \subseteq \LL(\phi)$ is simply an inclusion check for regular languages, it can be implemented efficiently using finite automata.
Our implementation uses the automata library \tool{libFAUDES} \cite{faudes}.

However, this is just a sufficient criterion for redundancy.
For example, a learned clause might be redundant w.r.t.\ an original clause, but such redundancies cannot be detected using $\LL$.
To see this, note that we have $|\LL(\phi)| = 1$ if $\phi$ is an original clause, but $|\LL(\phi)| = \infty$ if $\phi$ is a learned clause.

For \ref{calc:covered}, we also check redundancy via $\LL$, but if $\vec{\phi}' = \phi'$, i.e., if $|\vec{\phi}'| = 1$, then we only apply \ref{calc:covered} if $\phi'$ is an original clause.
Then $\LL(\phi') \subseteq \LL(\phi)$ for some $\phi \neq \phi'$ implies that $\phi$ is a learned clause.
Hence, we have $\LL(\phi') \subset \LL(\phi)$, as $|\LL(\phi')| = 1 < |\LL(\phi)| = \infty$.
This is just a heuristic, as even $\LL(\phi') \subset \LL(\phi)$ just implies $\phi' \sqsubseteq \phi$, but not $\phi' \sqsubset \phi$.
To see this, consider an original clause $\phi = (\F(X) \implies \F(0))$.
Then $\LL(\phi) = \{\phi\}$, $\accel(\phi) \equiv_\AA \phi$ (but not necessarily $\accel(\phi) = \phi$, as $\accel(\phi)$ and $\phi$ might differ syntactically), and $\LL(\accel(\phi)) =
\LL(\phi)^+$.
So we have
  $\LL(\phi) \subset \LL(\accel(\phi))$ and $\phi \sqsubseteq \accel(\phi)$, but $\phi \not\sqsubset \accel(\phi)$.
This is uncritical for proving $\unsat$, but a potential soundness issue for proving
$\sat$, which is one reason why our current implementation cannot prove $\sat$.

\subsubsection*{Implementing {\normalfont \ref{calc:step}} and Blocked Clauses}
To find an active clause in \ref{calc:step}, we proceed as described before \Cref{def:redundancy}, i.e., we search for a suitable element of $\mbp(\Pi)$ ``on the fly''.
So we search for a clause $\phi \in \Pi$ whose body-literal unifies with the head-literal of $\resolve(\vec{\phi})$ using an mgu $\theta$.
Then we use an SMT solver to check whether
\begin{equation}
  \label{eq:impl-resolution}
  \tag{\sc Step--SMT}
  \theta(\cond(\resolve(\vec{\phi}))) \land \theta(\cond(\phi)) \land \bigwedge_{\mathclap{\pi \in B \cap \mbp(\phi)}} \neg\theta(\cond(\pi))
\end{equation}
is satisfiable, where $B$ is the last element of $\vec{B}$.
Here, we assume that $\resolve(\vec{\phi})$ and $\phi$ are variable disjoint (and thus the
$\mgu$ $\theta$ exists).
If we find a model $\sigma$ for \ref{eq:impl-resolution}, then we apply \ref{calc:step}
with $\phi|_{\mbp(\cond(\phi),\sigma)}$.
So to exclude blocked clauses, we do not use the redundancy check based on $\LL$ explained above, but we conjoin the negated conditions of certain blocked clauses to \ref{eq:impl-resolution}.
To see why we only consider blocked clauses from $\mbp(\phi)$,
consider the case that $B = \{ \pi \}$ is a singleton.
Note that both $\theta(\cond(\phi))$ and $\theta(\cond(\pi))$ might contain variables that
do not occur as arguments of predicates in the (unified) head- or
body-literals.
So if
\begin{alignat*}{3}
  \phi & {} \equiv_\AA \forall \vec{X},\vec{Y}_{\phi},\vec{X}'.\ && \F(\vec{X}) \land \psi_{\phi}(\vec{X},\vec{Y}_{\phi},\vec{X}') && {} \implies \G(\vec{X}'), \\
  \phi' & {} \equiv_\AA \forall \vec{X},\vec{Y}_{\phi'},\vec{X}'.\ && \F(\vec{X}) \land \psi_{\phi'}(\vec{X},\vec{Y}_{\phi'},\vec{X}') && {} \implies \G(\vec{X}'), \qquad \text{and}\\
  \pi & {} \equiv_\AA \forall \vec{X},\vec{Y}_\pi,\vec{X}'.\ && \F(\vec{X}) \land \psi_\pi(\vec{X},\vec{Y}_\pi,\vec{X}') && {} \implies \G(\vec{X}'),
\end{alignat*}
for some $\phi' \in \mbp(\phi)$, then $\phi' \sqsubseteq \pi$ iff
\begin{equation}
  \label{eq:blocked}
  \tag{\protect{\ensuremath{{\sqsubseteq}\textsc{--equiv}}}}
  \models_\AA \psi_{\phi'} \implies \exists \vec{Y}_\pi.\ \psi_\pi.
\end{equation}
Thus, to ensure that we only find models $\sigma$ such that $\mbp(\cond(\phi),\sigma)$ is not blocked by $\pi$, we would have to conjoin
\[
  \neg(\psi_{\phi} \implies \exists \vec{Y}_\pi.\ \psi_\pi) \quad \equiv_\AA \quad \psi_\phi \land \forall \vec{Y}_{\pi}.\ \neg\psi_\pi
\]
to the SMT problem.
Unfortunately, as SMT solvers have limited support for quantifiers, such an encoding is impractical.
Hence, we again use a sufficient criterion for redundancy:
If
\begin{equation}
  \label{eq:blocked-approx}
  \tag{\protect{\ensuremath{{\sqsubseteq}\textsc{--sufficient}}}}
  \models_\AA \psi_{\phi'} \implies \psi_\pi,
\end{equation}
then \ref{eq:blocked} trivially holds as well.
So to exclude conjunctive variants $\phi'$ of $\phi$ where \ref{eq:blocked-approx} is valid, we add
\begin{equation}
  \label{eq:blocked-smt}
  \tag{\protect{\ensuremath{{\centernot{\sqsubseteq}}\textsc{--sufficient}}}}
  \neg(\psi_\phi \implies \psi_\pi) \quad \equiv_\AA \quad \psi_\phi \land \neg\psi_\pi
\end{equation}
to the SMT problem.
If $\vec{Y}_\pi \not\subseteq \vec{Y}_{\phi}$,
then satisfiability of \ref{eq:blocked-smt} is usually trivial.
Thus, to avoid increasing the size of the SMT problem unnecessarily, we only add \ref{eq:blocked-smt} to
the SMT problem if $\pi \in \mbp(\phi)$.
Instead, we could try to rename variables from $\vec{Y}_{\pi}$ to enforce $\vec{Y}_\pi \subseteq \vec{Y}_{\phi}$.
However, it is difficult to predict which renaming is the ``right'' one, i.e., which renaming would allow us to prove redundancy.

If $B \cap \mbp(\phi)$ contains several clauses $\pi_1,\ldots,\pi_\ell$, then \ref{eq:blocked-approx} becomes
\begin{equation}
  \label{eq:blocked-approx-set}
  \tag{\protect{\ensuremath{{\sqsubseteq}\textsc{--sufficient}^+}}}
  \models_\AA \psi_{\phi'} \implies \cond(\pi_1)\ \text{ or }\ \ldots \text{ or }\ \models_\AA \psi_{\phi'} \implies \cond(\pi_\ell)\
\end{equation}
Instead, our encoding excludes syntactic implicants $\phi'$ of $\phi$ where
\begin{equation}
  \label{eq:blocked-approx-necessary}
  \tag{\protect{\ensuremath{{\sqsubseteq}\textsc{--insufficient}^+}}}
  \models_\AA \psi_{\phi'} \implies \cond(\pi_1) \lor \ldots \lor \cond(\pi_\ell)
\end{equation}
which is a necessary, but not a sufficient condition for \ref{eq:blocked-approx-set}.
To see why this is not a problem, first note that \ref{eq:blocked-approx-set} trivially holds if $\psi_{\phi'} \in \{\cond(\pi_i) \mid 1 \leq i \leq \ell\}$.
Otherwise, we have
\[
\models_\AA \left(\cond(\pi_1) \lor \ldots \lor \cond(\pi_\ell) \right) \implies 
\bigvee \mbp(\psi_\phi) \setminus \{\psi_{\phi'}\}
\]
 because we assumed $\psi_{\phi'} \notin \{\cond(\pi_i) \mid 1 \leq i \leq \ell\}$, which
 implies $\{\cond(\pi_i) \mid 1 \leq i \leq \ell\} \subseteq \mbp(\psi_\phi) \setminus
 \{\psi_{\phi'}\}$.
 Together with \ref{eq:blocked-approx-necessary}, this implies
 \[ \models_\AA  \psi_{\phi'}  \implies 
 \bigvee \mbp(\psi_\phi) \setminus \{\psi_{\phi'}\}.\]
Therefore, we have $\psi_\phi \equiv_\AA \bigvee \mbp(\psi_\phi) \setminus \{\psi_{\phi'}\}$.
Thus, we may assume that \ref{eq:blocked-approx-necessary} implies redundancy without loss
of generality.
The reason is that we could analyze the following equivalent CHC problem instead of $\Pi$, otherwise:
\[
  (\Pi \setminus \{\phi\}) \cup (\mbp(\phi) \setminus \{\phi'\})
\]

Hence, in \ref{eq:impl-resolution}, we add (a variable-renamed variant of)
\begin{align*}
  \neg(\psi_{\phi} \implies \cond(\pi_1) \lor \ldots \lor \cond(\pi_\ell)) \quad \equiv_\AA \quad & \psi_\phi \land \neg\cond(\pi_1) \land \ldots \land \neg\cond(\pi_\ell) \\
  \equiv_\AA \quad & \cond(\phi) \land \bigwedge_{\mathclap{\pi \in B \cap \mbp(\phi)}} \neg\cond(\pi)
\end{align*}
to the SMT problem.

\begin{example}
  Consider the state \eqref{eq:after-covered} from \Cref{ex:calc}.
  First applying \ref{calc:step} with $\ref{eq:ex1-rule}|_{\psi_1}$ and then applying \ref{calc:covered} yields
  \[
    (\Pi_1,\vec{\phi}, [\emptyset,\{\ref{eq:ex1-rule}|_{\psi_1}\},\{\ref{eq:ex1-rule}|_{\psi_1}\}])
  \]
  where $\vec{\phi} = [\ref{eq:ex1-fact},\ref{eq:accel1}]$.
  When attempting a \ref{calc:step} with an element of $\mbp(\ref{eq:ex1-rule})$, we get:

  \noindent
  \resizebox{\textwidth}{!}{
    \begin{minipage}{\textwidth}
    \begin{align*}
      \theta(\cond(\vec{\phi})) & {} \equiv_\AA X_1 = 0 \land X_2 = 5k \land N > 0
\land X_1' < 5001
      \land X'_1 = X_1 + N \land X'_2 = X_2 \\
       \theta(\cond(\ref{eq:ex1-rule})) & {} \equiv_\AA ((X'_1 < 5k \land Y_2 = X'_2) \lor (X'_1 \geq 5k \land Y_2 = X'_2 + 1)) \\
      \hspace{1.3em} \bigwedge_{\mathclap{\pi \in B \cap \mbp(\ref{eq:ex1-rule})}} \neg
      \theta(\cond(\pi)) & {} = \neg \theta(\cond(\ref{eq:ex1-rule}|_{\psi_1})) \equiv_\AA
      X'_1 \geq 5k \lor  Y_2 \neq X'_2
    \end{align*}
    \end{minipage}
  }

  \smallskip
  \noindent
  Here, $5k$ abbreviates $5000$.
  Then \ref{eq:impl-resolution} is equivalent to
  \[
    X_1 = 0 \land X_2 = N = X'_1 = X'_2 = 5k \land Y_2 = 5001.
  \]
  Hence, we have $\sigma \models_\AA X'_1 \geq 5k \land Y_2 = X'_2 + 1$ for the unique model $\sigma$ of \ref{eq:impl-resolution}, i.e., $\sigma$ satisfies the second disjunct of $\cond(\ref{eq:ex1-rule})$.
  Thus, we add $\ref{eq:ex1-rule}|_{\mbp(\cond(\ref{eq:ex1-rule}), \sigma)} = \ref{eq:ex1-rule}|_{\psi_2}$ to the
  trace.
\end{example}

\subsubsection*{Leveraging Incremental SMT} The search for suitable models can naturally be implemented via incremental SMT solving:
When trying to apply \ref{calc:step}, we construct $\theta$ in such a way that $\theta(\cond(\resolve(\vec{\phi}))) = \cond(\resolve(\vec{\phi}))$.
This is easily possible, as $\theta$ just needs to unify predicates whose arguments are duplicate free and pairwise disjoint vectors of variables.
Then we push
\begin{equation}
  \label{eq:impl-resolution2}
  \tag{\sc Incremental}
  \theta(\cond(\phi)) \land \bigwedge_{\pi \in B \cap \mbp(\phi)} \neg\theta(\cond(\pi))
\end{equation}
to the SMT solver.
If the model from the previous resolution step can be extended to satisfy
\ref{eq:impl-resolution2}, then the SMT solver can do so, otherwise it searches for another model.
If it fails to find a model, we pop \ref{eq:impl-resolution2}, i.e., we remove it from the current SMT problem.
\ref{calc:accelerate} can be implemented similarly by popping $\theta(\cond(\vec{\phi}^\circlearrowleft))$ and pushing $\theta(\cond(\phi))$ instead.

Note that satisfiability of $\vec{\phi}$ is an invariant of ADCL.
Hence, as soon as the last element of $\vec{\phi}$ is a query, \ref{calc:refute} can be applied without further SMT checks.
Otherwise, if \ref{calc:step} cannot be applied with any clause, then \ref{calc:backtrack} or \ref{calc:accept} can be applied without further SMT queries.

\subsubsection*{Dealing with Incompleteness} As mentioned in \Cref{sec:ADCL}, we assumed that we have oracles for checking redundancy, satisfiability of $\QF(\AA)$-formulas, and acceleration when we formalized ADCL.
As this is not the case in practice, we now explain how to proceed if those techniques fail or approximate.

As explained above, SMT is needed for checking activity in \ref{calc:step}.
If the SMT solver fails, we assume inactivity.
Thus, we do not exhaust the entire search space if we falsely classify active clauses as inactive.
Hence, we may miss refutations, which is another reason why our current implementation
cannot prove $\sat$.

Regarding acceleration, our implementation of $\accel$ may return under-ap\-prox\-i\-ma\-tions, i.e., we just have $\ground(\accel(\phi)) \subseteq \bigcup_{n \in \NN_{\geq 1}} \ground(\phi^n)$.
While this is uncritical for correctness by itself (as learned clauses are still entailed
by $\Phi$), it weakens our heuristic for redundancy via $\LL$, as we no longer have
$\ground(\phi) = \ground(\LL(\phi))$, but just $\ground(\phi) \subseteq
\ground(\LL(\phi))$ for learned clauses $\phi$.

Another pitfall when using under-approximating acceleration techniques is that we may have $\vec{\phi}^\circlearrowleft \not\sqsubseteq \accel(\vec{\phi}^\circlearrowleft)$.
In this case, applying \ref{calc:accelerate} can result in an inconsistent trace where $\cond(\vec{\phi})$ is unsatisfiable.
To circumvent this problem, we only add $\accel(\vec{\phi}^\circlearrowleft)$ to the trace after removing $\vec{\phi}^\circlearrowleft$ if doing so results in a consistent trace.
Here, we could do better by taking the current model $\sigma$ into account when accelerating $\vec{\phi}^\circlearrowleft$ in order to ensure $\sigma \models_\AA \cond(\accel(\vec{\phi}^\circlearrowleft))$.
We leave that to future work.

\subsubsection*{Restarts} When testing our implementation, we noticed that several instances ``jiggled'', i.e., they were solved in some test runs, but failed in others.
The reason is a phenomenon that is well-known in SAT solving, called ``heavy-tail behavior''.
Here, the problem is that
the solver sometimes gets ``stuck'' in a part of the state space whose exploration is very
expensive, even though finding a solution in another part of the search space is well
within the solver's capabilities.
This problem also occurs in our implementation,
due to the depth-first strategy of our solver (where derivations may even be
non-terminating, see \Cref{thm:non-termination}).
To counter this problem, SAT solvers use restarts \cite{restarts}, where one of the most popular approaches has been proposed by Luby et al.\ \cite{luby}.
For SAT solving, the idea is to restart the search after a certain number of conflicts, where the number of conflicts for the next restart is determined by the \emph{Luby sequence}, scaled by a parameter $u$.
When restarting, randomization is used to avoid revisiting the same part of the search space again.
We use the same strategy with $u=10$, where we count the number of learned clauses instead of the number of conflicts.
To restart the search, we clear the trace, change the seed of the SMT solver (which
may result in different models
such that we may use
different syntactic implicants), and shuffle
the vectors of clauses (to change the order in which clauses are used for
resolution).

%% file: experiments.tex
\section{Related Work and Experiments}
\label{sec:experiments}

We presented the novel ADCL calculus for (dis)proving satisfiability of CHCs.
Its distinguishing feature is its use of acceleration for learning new clauses.
For unsatisfiability, these learned clauses often enable very short resolution proofs for CHC problems whose original clauses do not admit short resolution proofs.
For satisfiability, learned clauses often allow for covering the entire (usually infinite) search space by just considering finitely many resolution sequences.

\subsubsection*{Related Work}

The most closely related work is \cite{iosif12}, where acceleration is used in two ways: (1) as preprocessing and (2) to generalize interpolants in a CEGAR loop.
In contrast to (1), we use acceleration ``on the fly'' to accelerate resolvents.
In contrast to (2), we do not use abstractions, so our learned clauses can directly be used in resolution proofs.
Moreover, \cite{iosif12} only applies acceleration to conjunctive clauses, whereas we accelerate conjunctive variants of arbitrary clauses.
So in our approach, acceleration techniques are applicable more often, which is particularly useful for finding long counterexamples.
However, our approach \emph{solely} relies on acceleration to handle recursive CHCs, whereas \cite{iosif12} incorporates acceleration techniques into a CEGAR loop, which can also analyze recursive CHCs without accelerating them.
Thus, the approach from \cite{iosif12} is orthogonal to ADCL.
Both (1) and (2) are implemented in \tool{Eldarica}, but according to its authors, (2) is just supported for transition systems, but not yet for CHCs.
Hence, we only considered (1) in our evaluation (named \tool{Eld.} Acc.\ below).
Earlier, an approach similar to (2) has been proposed in \cite{caniart08}, but to the best of our knowledge, it has never been implemented.

\emph{Transition power abstraction} (TPA) \cite{golem} computes a sequence of over-ap\-prox\-i\-ma\-tions for transition systems where the $n^{th}$ element captures $2^n$ instead of just $n$ steps of the transition relation.
So like ADCL, TPA can help to find long refutations quickly, but in contrast to ADCL, TPA
relies on over-approximations.

Some leading techniques for CHC-SAT like GPDR \cite{GPDR} and, in particular, the
\tool{Spacer} algorithm \cite{spacer}, are adaptions of the \tool{IC3}
algorithm \cite{ic3} from transition systems to CHCs.
\tool{IC3} computes a sequence of abstractions of reachable states, aiming to find an abstraction that is inductive w.r.t.\ the transition relation and implies safety.

Other approaches for CHC-SAT are based on interpolation \cite{eldarica,ultimate-chc}, CEGAR and predicate abstraction \cite{eldarica,GrebenshchikovLPR12}, automata \cite{ultimate-chc}, machine learning \cite{freqhorn,synthhorn}, bounded model checking (BMC) \cite{bmc}, or combinations thereof.

Related approaches for transition systems include \cite{fast} and \cite{flata}.
The approach of \cite{fast} uses acceleration to analyze a sequence of \emph{flattenings} of a given transition system, i.e., under-approximations without nested loops,  until a counterexample is found or a fixpoint is reached.
Like ADCL, this approach does not terminate in general.
However, it does terminate for so-called \emph{flattable} systems.
Whether ADCL terminates for flattable systems as well is an interesting question for future work.
In contrast to ADCL, \cite{fast} has no notion of learning or redundancy, so that the same computations may have to be carried our several times for different flattenings.

The technique of \cite{flata} also lifts acceleration techniques to transition systems, but circumvents non-termination by using approximative acceleration techniques in the presence of disjunctions.
In contrast, ADCL handles disjunctions via syntactic implicants.
Like ADCL, \cite[Alg.~2]{flata} learns new transitions (Line 9), but only if they are non-redundant (Line 8).
However, it applies acceleration to all syntactic self-loops, whereas ADCL explores the
state space starting from facts, such that only reachable loops are accelerated.
Note that the approach from \cite{flata} is very similar to the approach that has been used by earlier versions of \tool{LoAT} for proving non-termination \cite{loat}.
We recently showed in \cite{adcl-nt} that for the purpose of proving non-termination,
ADCL is superior to \tool{LoAT}'s earlier approach.

Finally, \cite{underapprox15} uses under-approximating acceleration techniques to enrich
the control-flow graph of \pl{C} programs in order to find ``deep counterexamples'', i.e., long refutations.
In contrast to ADCL, \cite{underapprox15} relies on external model checkers for finding
counterexamples, and it has no notion of redundancy so that the model checker may explore
``superfluous'' paths that use original instead of accelerated edges of the control-flow graph.

Regarding acceleration, there are many results regarding classes of loops over integer variables where linear arithmetic suffices to express their transitive closure, i.e., they can be accelerated within a decidable theory.
The most important such classes are Difference Bounds \cite{differenceBounds}, Octagons \cite{bozga09a}, Finite Monoid Affine Relations \cite{finkel02}, and Vector Addition Systems with States \cite{vass}.
In an orthogonal line of research, monotonicity-based acceleration techniques have been developed \cite{underapprox15,fmcad19,acceleration-calculus}.
While the latter provide fewer theoretical guarantees in terms of completeness and whether the result can be expressed in a decidable logic or not, they are not restricted to loops whose transitive closure is definable in linear arithmetic.

Regarding other theories, the technique from \cite{vass} for Vector Addition Systems with States has also been applied to systems over rationals \cite{silverman19}.
Similarly, monotonicity-based approaches immediately carry over to rationals or reals.
The only approach for acceleration in the presence of Boolean variables that we are aware of is \cite{SchrammelJ11}.
However, this technique yields over-approximations.

Finally, some acceleration techniques for arrays have been proposed, e.g., \cite{aligators,accel-arr}.
The approach of \cite{accel-arr} improves the framework of \cite{vampire-arr} to reason
about programs with arrays using a first-order theorem prover by integrating specialized
techniques for dealing with array accesses where the indices are monotonically increasing or decreasing.
The technique of \cite{aligators} uses quantifier elimination techniques to accelerate loops where arrays can be separated into \emph{read-} and \emph{write-only} arrays.

\subsubsection*{Experiments}

So far, our implementation of ADCL in \tool{LoAT} is restricted to integer arithmetic.
Thus, to evaluate our approach, we used the examples from the category LIA-Lin (linear
CHCs with linear integer arithmetic) from the CHC competition~'22 \cite{CHC-COMP}, which
contains numerous CHC problems resulting from actual program
verification tasks.
Somewhat surprisingly, these examples contain additional features like variables of type
{\tt Bool} and the operators {\tt div} and {\tt mod}.
Since variables of type {\tt Bool} are used in most of the examples, we extended our implementation with rudimentary support for {\tt Bool}s.
In particular, we implemented a simplistic acceleration technique for {\tt Bool}s (note
that we cannot use the approach of \cite{SchrammelJ11}, as it yields over-approximations).
We excluded the $72$ examples that use {\tt div} or {\tt mod}, as those operators are not supported by our implementation.

To accelerate CHCs where some variables are of type {\tt Bool}, we use an adaption of the acceleration calculus from \cite{acceleration-calculus}.
To apply it to $\phi \Def (\F(\vec{X}) \land \psi \implies \F(\vec{Y}))$, $\phi$ needs to
be \emph{deterministic}, i.e., there must be a substitution $\theta$ such that
$\psi
\models_\AA \vec{Y} = \theta(\vec{X})$ and $\VV(\theta(\vec{X})) \subseteq
\vec{X}$.
Then \tool{LoAT} has to compute a \emph{closed form}, i.e., a vector $\vec{C}$ such that $\vec{C} \equiv_\AA \theta^N(\vec{X})$.
For integer variables, closed forms are computed via recurrence solving.
For Boolean variables $B$, \tool{LoAT} can only construct a closed form if there is a $k \in \NN$ such that $\theta^k(B)$ does not contain Boolean variables, or $\theta^k(B) = \theta^{k+1}(B)$.
Once a closed form has been computed, the calculus from \cite{acceleration-calculus} can be applied.
However, in the presence of Booleans, it has to be restricted to theory-agnostic acceleration techniques.
So more precisely, in the presence of Booleans, only the acceleration techniques \emph{monotonic increase} and \emph{monotonic decrease} from \cite{acceleration-calculus} can be used.

Using the remaining $427$ examples, we compared our implementation with the leading
CHC-SAT solvers \tool{Spacer}~\cite{spacer} (which is part of \tool{Z3}~\cite{z3}),
\tool{Eldarica}~\cite{eldarica}, and \tool{Golem}~\cite{golem}.
Additionally, we compared with \tool{Z3}'s implementation of BMC.
As mentioned above, \tool{Eldarica} supports acceleration as preprocessing.
Thus, besides \tool{Eldarica}'s default configuration (which does not use acceleration), we also compared with a configuration \tool{Eld.} Acc.\ where we enabled this feature.
By default, \tool{Golem} uses the \tool{Spacer} algorithm, but
the \tool{Spacer} implementation in \tool{Z3} performed better
in our tests.
Thus, we used \tool{Golem}'s implementation of TPA instead, which targets similar classes
of examples like ADCL, as explained above.
We used \tool{Z3} 4.11.2, \tool{Eldarica} 2.0.8, and \tool{Golem} 0.3.0 and ran our experiments on \tool{StarExec} \cite{starexec} with a wallclock timeout of $300$s, a cpu timeout of $1200$s, and a memory limit of $128$GB per example.

The results can be seen in \Cref{tab1}.
We evaluated all tools on the $209$ examples that do not use {\tt Bool}s (\textsf{Int only}) and on the entire benchmark set (\textsf{Int \& Bool}).
The table on the left shows that \tool{LoAT} is very competitive w.r.t.\ proving $\unsat$ in terms of solved instances.
The entries in the column ``unique'' show the number of examples where the respective tool
succeeds and all others fail. Here we disregard \tool{Eld.} Acc., as it would be pointless to consider
several variants of the same algorithm in such a comparison.
If we consider \tool{Eld.} Acc.\ instead of \tool{Eldarica}, then the numbers change
according to the values given in parentheses.

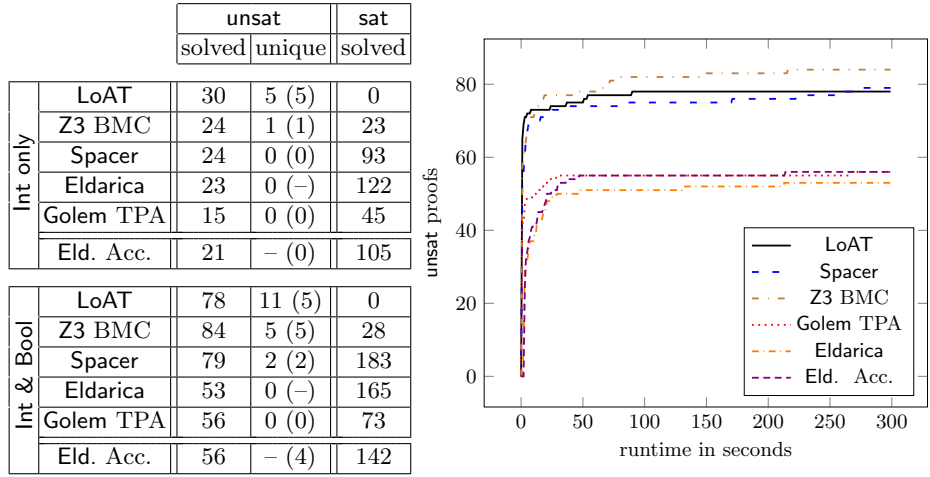
\begin{figure}[t]
  \begin{minipage}{0.438\textwidth}
    \begin{tabular}{|c|c||c|c||c|}
      \hhline{~~---} \multicolumn{1}{c}{} & \multicolumn{1}{c|}{} & \multicolumn{2}{c||}{$\unsat$} & $\sat$ \\
      \hhline{~~---} \multicolumn{1}{c}{} & \multicolumn{1}{c|}{} & solved & unique & solved \\
      \hhline{~~---} \multicolumn{2}{c}{} \\[-0.5em] \hline

      \multirow{4}{*}{\rotatebox{90}{\sf Int only \hspace{1.75em}}} & \tool{LoAT} & 30 & 5 (5) & 0 \\
      \hhline{~----} & \tool{Z3} BMC & 24 & 1 (1) & 23 \\
      \hhline{~----} & \tool{Spacer} & 24 & 0 (0) & 93 \\
      \hhline{~----} & \tool{Eldarica} & 23 & 0 (--) & 122 \\
      \hhline{~----} & \tool{Golem} TPA & 15 & 0 (0) & 45 \\
      \hhline{|~|====} & \tool{Eld.} Acc. & 21 & -- (0) & 105 \\
      \hline \multicolumn{2}{c}{} \\[-0.5em]

      \hline \multirow{4}{*}{\rotatebox{90}{\sf Int \& Bool \hspace{1em}}} & \tool{LoAT} & 78 & 11 (5) & 0 \\
      \hhline{~----} & \tool{Z3} BMC & 84 & 5 (5) & 28 \\
      \hhline{~----} & \tool{Spacer} & 79 & 2 (2) & 183 \\
      \hhline{~----} & \tool{Eldarica} & 53 & 0 (--) & 165 \\
      \hhline{~----} & \tool{Golem} TPA & 56 & 0 (0) & 73 \\
      \hhline{|~|====} & \tool{Eld.} Acc. & 56 & -- (4) & 142 \\
      \hline
    \end{tabular}
  \end{minipage}
  \begin{minipage}{0.557\textwidth}
    \vspace{1em}
    \begin{tikzpicture}[scale=0.86]
      \begin{axis}[
        legend pos=south east,
        xlabel=runtime in seconds,
        ylabel=$\unsat$ proofs,]
        \addplot[color=black,solid,thick] table[col sep=comma,header=false,x index=0,y index=1] {loat_rt_all.csv};
        \addlegendentry{\tool{LoAT}}
        \addplot[color=blue,loosely dashed,thick] table[col sep=comma,header=false,x index=0,y index=1] {spacer_rt_all.csv};
        \addlegendentry{\tool{Spacer}}
        \addplot[color=brown,loosely dashdotted,thick] table[col sep=comma,header=false,x index=0,y index=1] {bmc_rt_all.csv};
        \addlegendentry{\tool{Z3} BMC}
        \addplot[color=red,dotted,thick] table[col sep=comma,header=false,x index=0,y index=1] {golem_rt_all.csv};
        \addlegendentry{\tool{Golem} TPA}
        \addplot[color=orange,dashdotted,thick] table[col sep=comma,header=false,x index=0,y index=1] {eld_rt_all.csv};
        \addlegendentry{\tool{Eldarica}}
        \addplot[color=violet,densely dashed,thick] table[col sep=comma,header=false,x index=0,y index=1] {eld_stac_rt_all.csv};
        \addlegendentry{\tool{Eld.\ } Acc.}
      \end{axis}
    \end{tikzpicture}
  \end{minipage}
  \caption{Comparing \tool{LoAT} with other CHC solvers}
  \label{tab1}
\end{figure}

\begin{table}[h]
  \begin{center}
    \begin{tabular}{|c|c|c|}
      \hline
      example & \tool{LoAT}'s refutation & original refutation \\ \hline \hline
      {\tt chc-LIA-Lin\_043.smt2} & 6 & 965553 \\ \hline
      {\tt chc-LIA-Lin\_045.smt2} & 2 & 684682683 \\ \hline
      {\tt chc-LIA-Lin\_047.smt2} & 3 & 72536 \\ \hline
      {\tt chc-LIA-Lin\_059.smt2} & 3 & 100000001 \\ \hline
      {\tt chc-LIA-Lin\_154.smt2} & 2 & 134217729 \\ \hline
      {\tt chc-LIA-Lin\_358.smt2} & 12 & 400005 \\ \hline
      {\tt chc-LIA-Lin\_362.smt2} & 12 & 400005 \\ \hline
      {\tt chc-LIA-Lin\_386.smt2} & 15 & 600003 \\ \hline
      {\tt chc-LIA-Lin\_401.smt2} & 8 & 200005 \\ \hline
      {\tt chc-LIA-Lin\_402.smt2} & 4 & 134217723 \\ \hline
      {\tt chc-LIA-Lin\_405.smt2} & 9 & 100012 \\ \hline
    \end{tabular}
  \end{center}
  \caption{Comparing lengths of refutation}\vspace*{-1cm}
  \label{tab2}
\end{table}

The numbers indicate that \tool{LoAT} is particularly powerful on examples that operate on {\tt Int}s only, but it is also competitive for proving unsatisfiability of examples that may operate on {\tt Bool}s, where it is only slightly weaker than \tool{Spacer} and \tool{Z3} BMC.
This is not surprising, as the core of \tool{LoAT}'s approach are its acceleration techniques, which have been designed for integers.
In contrast, \tool{Spacer}'s algorithm is similar to GPDR \cite{GPDR}, which generalizes
the \tool{IC3} algorithm \cite{ic3} from transition systems over Booleans to transition systems over theories (like integers), and BMC is theory agnostic.

The figure on the right shows how many proofs of unsatisfiability were found within a given runtime, by each tool.
Here, all examples (\textsf{Int \& Bool}) are taken into account.
\tool{LoAT} finds many proofs of unsatisfiability quickly (73 proofs within 8s).
\tool{Z3} BMC catches up after 12s (73 proofs for both \tool{LoAT} and \tool{Z3} BMC) and
takes over the lead after 14s (\tool{LoAT} 73, \tool{Z3} BMC 74).
\tool{Spacer} catches up with \tool{LoAT} after
260s.

To illustrate \tool{LoAT}'s ability to find short refutations, \Cref{tab2} compares the number of resolution steps in \tool{LoAT}'s ``accelerated'' refutations (that also use learned clauses) with the corresponding refutations that only use original clauses.
Here, we restrict ourselves to those instances that can only be solved by \tool{LoAT}, as
the unsatisfiable CHC problems that can also be solved by other tools usually already
admit quite short refutations without learned clauses.
To compute the length of the original refutations, we instrumented each predicate with an additional argument $\mathit{c}$.
Moreover, we extended the condition of each fact $\psi \implies \G(..., \mathit{c})$ with $\mathit{c} = 1$ and the condition of each rule $\F(..., \mathit{c}) \land \psi \implies \G(..., \mathit{c}')$ with $\mathit{c}' = \mathit{c} + 1$.
Then the value of $\mathit{c}$ before applying a query corresponds to the number of
resolution steps that one would need if one only used original clauses, and it can be extracted from the model found by the SMT solver.
The numbers clearly show that learning clauses via acceleration allows to reduce the length of refutations dramatically.
In $76$ cases, \tool{LoAT} learned clauses with non-linear arithmetic.

Our implementation is open-source and available on Github.
For the sources, a pre-compiled binary, and more information on our evaluation, we refer to \cite{artifact,website}.
In future work, we plan to extend our implementation to also prove $\sat$, and we will investigate how to construct models for satisfiable CHC problems.
Moreover, we want to add support for further theories by developing specialized acceleration techniques.
Furthermore, we intend to lift ADCL to non-linear CHCs.

%% file: main.bbl
\providecommand{\noopsort}[1]{}

%% file: proofs.tex
\begin{appendix}

  \section{Additional Definitions}

  Here, we provide two additional definitions that are used in the proofs.

  \begin{definition}[Resolution with Ground Instances]
    \label{def:resolution-ground}
    Let $\phi$ and $\phi'$ be ground instances of CHCs.\footnote{Note that in our setting, in general a ground instance of a CHC is not a CHC, because here the predicates are not applied to variables anymore.
    }
    If
    \[
      \phi = (\eta \implies \F(\vec{s})) \quad \text{and} \quad \phi' = (\F(\vec{s}) \implies \eta'),
    \]
 \begin{align*}
    \text{then} && \resolve(\phi,\phi') & {} \Def \eta \implies \eta'. \\
    \text{Otherwise,} && \resolve(\phi,\phi') & {} \Def (\bot \implies \bot).
 \end{align*}
 Here, $\eta$ can also be $\top$ and $\eta'$ can also be $\bot$.
    For non-empty sequences of ground instances, $\resolve$ is defined analogously to sequences of CHCs.
  \end{definition}

  \begin{definition}[Liftings of Resolution]
  \label{def:resolution-lifted}
    We define:
  \begin{align*}
    \resolve(\LL) & {} \Def \{ \resolve(\vec{\pi}) \mid \vec{\pi} \in \LL \} & \text{for sets of sequences of CHCs $\LL$} \\
    \resolve(\Phi_1,\Phi_2) & {} \Def \{ \resolve(\phi_1,\phi_2) \mid \phi_1 \in \Phi_1, \mathrlap{\phi_2 \in \Phi_2 \}} & \text{for sets of CHCs $\Phi_1,\Phi_2$} \\
    \resolve(\LL_1,\LL_2) & {} \Def \resolve(\resolve(\LL_1),\resolve(\LL_2)) & \text{for sets of sequences of CHCs $\LL_1,\LL_2$}
  \end{align*}
  For sets of (sequences of) ground instances of CHCs, we lift $\resolve$ analogously.
\end{definition}

  \section{Auxiliary Lemmas}\label{Auxiliary Lemmas}

This appendix contains a number of auxiliary lemmas that are needed for the proofs of our
main results. We start with a lemma which essentially corresponds to the classical lifting
lemma of resolution.

  \begin{lemma}[Resolution Distributes over $\ground$]
    \label{lem:distribute}
    For two CHCs $\phi_1, \phi_2$ we have:
    \[
      \ground(\resolve(\phi_1, \phi_2)) = \{\phi \in \resolve(\ground(\phi_1),\ground(\phi_2)) \mid \;\models_\AA \cond(\phi)\}
    \]
  \end{lemma}
  \begin{proof}
    If $\resolve(\phi_1, \phi_2) = (\bot \implies \bot)$, then $\resolve(\pi_1,\pi_2) = (\bot \implies \bot)$ for all ground instances $\pi_1$ of $\phi_1$ and $\pi_2$ of $\phi_2$.
    Thus,
    \[
      \ground(\resolve(\phi_1, \phi_2)) = \{\phi \in \resolve(\ground(\phi_1),\ground(\phi_2)) \mid \; \models_\AA \cond(\phi) \} = \emptyset.
    \]
    Assume $\resolve(\phi_1, \phi_2) \neq (\bot \implies \bot)$.
    Then we have $\phi_1 = (\eta_1 \land \psi_1 \implies \eta'_1)$, $\phi_2 = (\eta_2 \land \psi_2 \implies \eta'_2)$, and $\theta = \mgu(\eta'_1, \eta_2)$.
    Then
    \begin{align*}
      & \{\phi \in \resolve(\ground(\phi_1),\ground(\phi_2)) \mid \; \models_\AA \cond(\phi) \} \\
      {} = {} & \{\resolve(\pi_1, \pi_2) \mid \pi_1 \in \ground(\phi_1), \pi_2 \in \ground(\phi_2), \models_\AA \cond(\phi)\} \tag{\Cref{def:resolution-lifted}} \\
      {} = {} & \{\resolve(\phi_1\sigma_1, \phi_2\sigma_2) \mid \; \sigma_1 \models_\AA \psi_1, \sigma_2 \models_\AA \psi_2, \eta'_1\sigma_1 = \eta_2\sigma_2\} \tag{def.\ of $\ground$} \\
      {} = {} & \{\eta_1\sigma_1 \implies \eta'_2\sigma_2 \mid \sigma_1 \models_\AA \psi_1, \sigma_2 \models_\AA \psi_2, \eta'_1\sigma_1 = \eta_2\sigma_2\} \tag{\Cref{def:resolution-ground}} \\
      {} = {} & \{\eta_1\theta\sigma \implies \eta'_2\theta\sigma \mid \sigma \models_\AA \psi_1\theta, \sigma \models_\AA \psi_2\theta\} \tag{$\dagger$} \\
      {} = {} & \{\eta_1\theta\sigma \implies \eta'_2\theta\sigma \mid \sigma \models_\AA \psi_1\theta \land \psi_2\theta\} \\
      {} = {} & \ground((\eta_1 \land \psi_1 \land \psi_2 \implies \eta'_2)\theta) \tag{by def.\ of $\ground$}\\
      {} = {} & \ground(\resolve(\phi_1, \phi_2)) \tag{\Cref{def:resolution}}
    \end{align*}
    The step marked with $(\dagger)$ holds as $\theta$ is the most general unifier of
    $\eta'_1$ and $\eta_2$ and as $\phi_1$ and $\phi_2$ are assumed to be variable disjoint in
    \Cref{def:resolution}.
    Therefore, there exists a model $\sigma$ such that $x\sigma_1 = x\theta \sigma$ for all $x \in \VV(\phi_1)$ and $x\sigma_2 = x\theta \sigma$ for all $x \in \VV(\phi_2)$. \qed
  \end{proof}

  \begin{lemma}[Associativity of Resolution]
    \label{lem:assoc}
    \[
      \ground(\resolve(\resolve(\phi_1,\phi_2), \phi_3)) = \ground(\resolve(\phi_1,\resolve(\phi_2,\phi_3)))
    \]
  \end{lemma}
  \begin{proof}
    Let $\phi \in \ground(\resolve(\resolve(\phi_1,\phi_2), \phi_3))$.
    Due to \Cref{lem:distribute}, there are ground instances $\phi'_1,\phi'_2,\phi'_3$ of $\phi_1,\phi_2,\phi_3$ such that $\phi = \resolve(\resolve(\phi'_1,\phi'_2), \phi'_3)$.
    Since $\phi'_1,\phi'_2,\phi'_3$ are ground, there exist ground atoms $\eta_1, \ldots, \eta_4$ such that
    \begin{align*}
      & \resolve(\resolve(\phi'_1,\phi'_2), \phi'_3) \\
      {} = {} & \resolve(\resolve(\eta_1 \implies \eta_2,\eta_2 \implies \eta_3), \eta_3 \implies \eta_4) \\
      {} = {} & \resolve(\eta_1 \implies \eta_3, \eta_3 \implies \eta_4) \\
      {} = {} & (\eta_1 \implies \eta_4) \\
      {} = {} & \resolve(\eta_1 \implies \eta_2, \eta_2 \implies \eta_4) \\
      {} = {} & \resolve(\eta_1 \implies \eta_2, \resolve(\eta_2 \implies \eta_3, \eta_3 \implies \eta_4)) \\
      {} = {} & \resolve(\phi'_1, \resolve(\phi'_2, \phi'_3)) \tag*{\qed}
    \end{align*}
  \end{proof}

  \begin{lemma}[Composition of Ground Instances]
    \label{lem:ground}
    If $\ground(\vec{\phi}_1) \subseteq \ground(\vec{\pi}_1)$ and $\ground(\vec{\phi}_2) \subseteq \ground(\vec{\pi}_2)$, then
    \[
      \ground(\vec{\phi}_1::\vec{\phi}_2) \subseteq \ground(\vec{\pi}_1::\vec{\pi}_2).
    \]
  \end{lemma}
  \begin{proof}
    \begin{align*}
      & \ground(\vec{\phi}_1::\vec{\phi}_2)                                                                                                      \\
      {} = {}         & \ground(\resolve(\resolve(\vec{\phi}_1),\resolve(\vec{\phi}_2))) \tag{\Cref{lem:assoc}}                                                  \\
      {} = {}         & \{\phi \in \resolve(\ground(\vec{\phi}_1),\ground(\vec{\phi}_2)) \mid \; \models_\AA \cond(\phi)\} \tag{\Cref{lem:distribute}}           \\
      {} \subseteq {} & \{\phi \in \resolve(\ground(\vec{\pi}_1),\ground(\vec{\phi}_2)) \mid \; \models_\AA \cond(\phi)\} \tag{$\ground(\vec{\phi}_1) \subseteq \ground(\vec{\pi}_1)$} \\
      {} \subseteq {} & \{\phi \in \resolve(\ground(\vec{\pi}_1),\ground(\vec{\pi}_2)) \mid \; \models_\AA \cond(\phi)\} \tag{$\ground(\vec{\phi}_2) \subseteq \ground(\vec{\pi}_2)$} \\
      {} = {}         & \ground(\resolve(\vec{\pi}_1),\resolve(\vec{\pi}_2)) \tag{\Cref{lem:distribute}}                                                         \\
      {} = {}         & \ground(\vec{\pi}_1::\vec{\pi}_2) \tag{\Cref{lem:assoc}}
    \end{align*}
    \qed
  \end{proof}

  \begin{lemma}[Resolution Distributes over $\LL$]
    \label{lem:distribute2}
    If $\Phi \leadsto^* (\Pi,\vec{\phi},\vec{B},\LL)$, $\vec{\pi} \in \mbp(\Pi)^*$, and $\pi \in \mbp(\Pi)$, then
    \[
      \ground(\resolve(\LL(\vec{\pi}::\pi)))) = \ground(\resolve(\LL(\vec{\pi}),\LL(\pi))).
    \]
  \end{lemma}
  \begin{proof}
     \begin{align*}
      & \ground(\resolve(\LL(\vec{\pi}),\LL(\pi))) \\
      {} = {} & \ground(\resolve(\resolve(\LL(\vec{\pi})),\resolve(\LL(\pi))))
      \tag{\Cref{def:resolution-lifted}} \\
      {} = {} & \ground(\{\resolve(\pi_1, \pi_2) \mid \pi_1 \in \resolve(\LL(\vec{\pi})), \pi_2 \in \resolve(\LL(\pi))\}) \tag{\Cref{def:resolution-lifted}} \\
      {} = {} & \ground(\{\resolve(\vec{\pi}_1::\vec{\pi}_2) \mid \vec{\pi}_1 \in \LL(\vec{\pi}), \vec{\pi}_2 \in \LL(\pi)\}) \tag{\Cref{lem:assoc}} \\
      {} = {} & \ground(\{\resolve(\vec{\pi}) \mid \vec{\pi} \in \LL(\vec{\pi}) :: \LL(\pi)\}) \\
      {} = {} & \ground(\{\resolve(\vec{\pi}) \mid \vec{\pi} \in \LL(\vec{\pi}::\pi)\} \tag{by definition of $\LL$}) \\
      {} = {} & \ground(\resolve(\LL(\vec{\pi}::\pi))) \tag{\Cref{def:resolution-lifted}}
    \end{align*}
    \qed
  \end{proof}

\end{appendix}

\appendixproofsection{Missing Proofs}\label{sec:MissingProofs}
\appendixproof*{thm:soundness}
\appendixproof*{thm:normal forms}
\appendixproof*{thm:refutational-complete}
\appendixproof*{lem:redundant}
\appendixproof*{thm:non-termination}

%% file: main.bbl
\begin{thebibliography}{10}
\providecommand{\url}[1]{\texttt{#1}}
\providecommand{\urlprefix}{URL }
\providecommand{\doi}[1]{https://doi.org/#1}

\bibitem{artifact}
Artifact for ``{ADCL}: {A}cceleration {D}riven {C}lause {L}earning for
  {C}onstrained {H}orn {C}lauses'' (2023). \doi{10.5281/zenodo.8146788}

\bibitem{website}
Evaluation of ``{ADCL}: {A}cceleration {D}riven {C}lause {L}earning for
  {C}onstrained {H}orn {C}lauses'' (2023),
  \url{https://loat-developers.github.io/adcl-evaluation}, \linebreak source
  code of \tool{LoAT} available at
  \url{https://github.com/loat-developers/LoAT/tree/v0.4.0}

\bibitem{solcmc}
Alt, L., Blicha, M., Hyv{\"{a}}rinen, A.E.J., Sharygina, N.: \tool{SolCMC}:
  \pl{Solidity} compiler's model checker. In: CAV~'22. pp. 325--338. LNCS 13371
  (2022). \doi{10.1007/978-3-031-13185-1\_16}

\bibitem{purrs}
Bagnara, R., Pescetti, A., Zaccagnini, A., Zaffanella, E.: \tool{PURRS}:
  Towards computer algebra support for fully automatic worst{-}case complexity
  analysis. CoRR  \textbf{abs/cs/0512056} (2005).
  \doi{10.48550/arXiv.cs/0512056}

\bibitem{fast}
Bardin, S., Finkel, A., Leroux, J., Schnoebelen, P.: Flat acceleration in
  symbolic model checking. In: {ATVA} '05. pp. 474--488. LNCS 3707 (2005).
  \doi{10.1007/11562948\_35}

\bibitem{bmc}
Biere, A.: Bounded model checking. In: Handbook of Satisfiability - Second
  Edition, Frontiers in Artificial Intelligence and Applications, vol.~336, pp.
  739--764. {IOS} Press (2021). \doi{10.3233/FAIA201002}

\bibitem{golem}
Blicha, M., Fedyukovich, G., Hyv{\"{a}}rinen, A.E.J., Sharygina, N.: Transition
  power abstractions for deep counterexample detection. In: TACAS~'22. pp.
  524--542. LNCS 13243 (2022). \doi{10.1007/978-3-030-99524-9\_29}

\bibitem{bozga09a}
Bozga, M., G{\^{\i}}rlea, C., Iosif, R.: Iterating octagons. In: TACAS~'09. pp.
  337--351. LNCS 5505 (2009). \doi{10.1007/978-3-642-00768-2\_29}

\bibitem{bozga10}
Bozga, M., Iosif, R., Kone\v{c}n\'{y}, F.: Fast acceleration of ultimately
  periodic relations. In: CAV~'10. pp. 227--242. LNCS 6174 (2010).
  \doi{10.1007/978-3-642-14295-6\_23}

\bibitem{flata}
Bozga\noopsort{1}, M., Iosif, R., Kone{\v{c}}n{\'{y}}, F.: Relational analysis
  of integer programs. Tech. Rep. TR-2012-10, VERIMAG (2012),
  \url{https://www-verimag.imag.fr/TR/TR-2012-10.pdf}

\bibitem{ic3}
Bradley, A.R.: {SAT}-based model checking without unrolling. In: {VMCAI} '11.
  pp. 70--87. LNCS 6538 (2011). \doi{10.1007/978-3-642-18275-4\_7}

\bibitem{horndroid}
Calzavara, S., Grishchenko, I., Maffei, M.: \tool{HornDroid}: Practical and
  sound static analysis of {Android} applications by {SMT} solving. In:
  EuroS{\&}P~'16. pp. 47--62. {IEEE} (2016). \doi{10.1109/EuroSP.2016.16}

\bibitem{caniart08}
Caniart, N., Fleury, E., Leroux, J., Zeitoun, M.: Accelerating
  interpolation-based model-checking. In: TACAS\ '08. pp. 428--442. LNCS 4963
  (2008). \doi{10.1007/978-3-540-78800-3\_32}

\bibitem{CHC-COMP}
{CHC Competition}, \url{https://chc-comp.github.io}

\bibitem{accel-arr}
Chen, Y., Kov{\'{a}}cs, L., Robillard, S.: Theory-specific reasoning about
  loops with arrays using {{\textsf{{{V}ampire}}}}. In: Vampire@IJCAR '16. pp.
  16--32. EPiC 44 (2016). \doi{10.29007/qk21}

\bibitem{differenceBounds}
Comon, H., Jurski, Y.: Multiple counters automata, safety analysis and
  {P}resburger arithmetic. In: CAV~'98. pp. 268--279. LNCS 1427 (1998).
  \doi{10.1007/BFb0028751}

\bibitem{ultimate-chc}
Dietsch, D., Heizmann, M., Hoenicke, J., Nutz, A., Podelski, A.: \tool{Ultimate
  TreeAutomizer} {(CHC-COMP} tool description). In: HCVS/PERR@ETAPS '19. pp.
  42--47. EPTCS 296 (2019). \doi{10.4204/EPTCS.296.7}

\bibitem{yices}
Dutertre, B.: \tool{Yices} 2.2. In: CAV~'14. pp. 737--744. LNCS 8559 (2014).
  \doi{10.1007/978-3-319-08867-9\_49}

\bibitem{korn}
Ernst, G.: Loop verification with invariants and contracts. In: VMCAI~'22. pp.
  69--92. LNCS 13182 (2022). \doi{10.1007/978-3-030-94583-1\_4}

\bibitem{freqhorn}
Fedyukovich, G., Prabhu, S., Madhukar, K., Gupta, A.: Solving constrained
  {Horn} clauses using syntax and data. In: FMCAD~'18. pp.~1--9 (2018).
  \doi{10.23919/FMCAD.2018.8603011}

\bibitem{finkel02}
Finkel, A., Leroux, J.: How to compose {P}resburger{-}accelerations:
  Applications to broadcast protocols. In: FSTTCS~'02. pp. 145--156. LNCS 2556
  (2002). \doi{10.1007/3-540-36206-1\_14}

\bibitem{fmcad19}
Frohn, F., Giesl, J.: Proving non-termination via loop acceleration. In:
  FMCAD~'19. pp. 221--230 (2019). \doi{10.23919/FMCAD.2019.8894271}

\bibitem{acceleration-calculus}
Frohn, F.: A calculus for modular loop acceleration. In: TACAS~'20. pp. 58--76.
  LNCS 12078 (2020). \doi{10.1007/978-3-030-45190-5\_4}

\bibitem{loat}
Frohn, F., Giesl, J.: Proving non-termination and lower runtime bounds with
  \tool{LoAT} (system description). In: IJCAR~'22. pp. 712--722. LNCS 13385
  (2022). \doi{10.1007/978-3-031-10769-6\_41}

\bibitem{adcl-nt}
Frohn, F., Giesl, J.: Proving non-termination by {A}cceleration {D}riven
  {C}lause {L}earning. In: CADE '23. LNCS (2023), to appear. Full version
  appeared in CoRR \textbf{abs/2304.10166},
  \url{https://doi.org/10.48550/arXiv.2304.10166}.

\bibitem{iosif17}
Ganty, P., Iosif, R., Kone\v{c}n\'{y}, F.: Underapproximation of procedure
  summaries for integer programs. Int. J. Softw. Tools Technol. Transf.
  \textbf{19}(5),  565--584 (2017). \doi{10.1007/s10009-016-0420-7}

\bibitem{restarts}
Gomes, C.P., Selman, B., Kautz, H.A.: Boosting combinatorial search through
  randomization. In: {AAAI} '98. pp. 431--437 (1998),
  \url{https://www.cs.cornell.edu/gomes/pdf/1998\_gomes\_aaai\_iaai\_boosting.pdf}

\bibitem{GrebenshchikovLPR12}
Grebenshchikov, S., Lopes, N.P., Popeea, C., Rybalchenko, A.: Synthesizing
  software verifiers from proof rules. In: {PLDI} '12. pp. 405--416 (2012).
  \doi{10.1145/2254064.2254112}

\bibitem{seahorn}
Gurfinkel, A., Kahsai, T., Komuravelli, A., Navas, J.A.: The \tool{SeaHorn}
  verification framework. In: CAV~'15. pp. 343--361. LNCS 9206 (2015).
  \doi{10.1007/978-3-319-21690-4\_20}

\bibitem{vass}
Haase, C., Halfon, S.: Integer vector addition systems with states. In: {RP}
  '14. pp. 112--124. LNCS 8762 (2014). \doi{10.1007/978-3-319-11439-2\_9}

\bibitem{aligators}
Henzinger, T.A., Hottelier, T., Kov{\'{a}}cs, L., Rybalchenko, A.: Aligators
  for arrays (tool paper). In: {LPAR} '10. pp. 348--356. LNCS 6397 (2010).
  \doi{10.1007/978-3-642-16242-8\_25}

\bibitem{GPDR}
Hoder, K., Bj{\o}rner, N.S.: Generalized property directed reachability. In:
  {SAT} '12. pp. 157--171. LNCS 7317 (2012).
  \doi{10.1007/978-3-642-31612-8\_13}

\bibitem{iosif12}
Hojjat, H., Iosif, R., Kone\v{c}n\'{y}, F., Kuncak, V., R{\"{u}}mmer, P.:
  Accelerating interpolants. In: ATVA~'12. pp. 187--202. LNCS 7561 (2012).
  \doi{10.1007/978-3-642-33386-6\_16}

\bibitem{eldarica}
Hojjat, H., R{\"{u}}mmer, P.: The \tool{Eldarica} {Horn} solver. In: FMCAD~'18.
  pp.~1--7 (2018). \doi{10.23919/FMCAD.2018.8603013}

\bibitem{jayhorn}
Kahsai, T., R{\"{u}}mmer, P., Sanchez, H., Sch{\"{a}}f, M.: \tool{JayHorn}: {A}
  framework for verifying \pl{Java} programs. In: CAV~'16. pp. 352--358. LNCS
  9779 (2016). \doi{10.1007/978-3-319-41528-4\_19}

\bibitem{spacer}
Komuravelli, A., Gurfinkel, A., Chaki, S.: {SMT}-based model checking for
  recursive programs. Formal Methods Syst. Des.  \textbf{48}(3),  175--205
  (2016). \doi{10.1007/s10703-016-0249-4}

\bibitem{ringen}
Kostyukov, Y., Mordvinov, D., Fedyukovich, G.: Beyond the elementary
  representations of program invariants over algebraic data types. In: {PLDI}
  '21. pp. 451--465 (2021). \doi{10.1145/3453483.3454055}

\bibitem{vampire-arr}
Kov{\'{a}}cs, L., Voronkov, A.: Finding loop invariants for programs over
  arrays using a theorem prover. In: {FASE} '09. pp. 470--485. LNCS 5503
  (2009). \doi{10.1007/978-3-642-00593-0\_33}

\bibitem{underapprox15}
Kroening, D., Lewis, M., Weissenbacher, G.: Under{-}approximating loops in
  \pl{C} programs for fast counterexample detection. Formal Methods Syst. Des.
  \textbf{47}(1),  75--92 (2015). \doi{10.1007/s10703-015-0228-1}

\bibitem{faudes}
{libFAUDES Library}, \url{https://fgdes.tf.fau.de/faudes/index.html}

\bibitem{luby}
Luby, M., Sinclair, A., Zuckerman, D.: Optimal speedup of {Las Vegas}
  algorithms. Information Processing Letters  \textbf{47}(4),  173--180 (1993).
  \doi{10.1016/0020-0190(93)90029-9}

\bibitem{rusthorn}
Matsushita, Y., Tsukada, T., Kobayashi, N.: \tool{RustHorn}: {CHC}-based
  verification for \pl{Rust} programs. {ACM} Trans. Program. Lang. Syst.
  \textbf{43}(4),  15:1--15:54 (2021). \doi{10.1145/3462205}

\bibitem{z3}
\noopsort{Moura}{de Moura}, L., Bj{\o}rner, N.: \tool{Z3}: An efficient {SMT}
  solver. In: TACAS\ '08. pp. 337--340. LNCS 4963 (2008).
  \doi{10.1007/978-3-540-78800-3\_24}

\bibitem{oeis-thue-morse}
{OEIS Foundation Inc.}: Thue-{M}orse sequence. The {O}n-{L}ine {E}ncyclopedia
  of {I}nteger {S}equences, published electronically at
  \url{https://oeis.org/A010060}

\bibitem{oeis-thue-morse-difference}
{OEIS Foundation Inc.}: First differences of {T}hue-{M}orse sequence. The
  {O}n-{L}ine {E}ncyclopedia of {I}nteger {S}equences (1999), published
  electronically at \url{https://oeis.org/A029883}

\bibitem{SchrammelJ11}
Schrammel, P., Jeannet, B.: Logico-numerical abstract acceleration and
  application to the verification of data-flow programs. In: {SAS} '11. pp.
  233--248. LNCS 6887 (2011). \doi{10.1007/978-3-642-23702-7\_19}

\bibitem{silverman19}
Silverman, J., Kincaid, Z.: Loop summarization with rational vector addition
  systems. In: CAV~'19. pp. 97--115. LNCS 11562 (2019).
  \doi{10.1007/978-3-030-25543-5\_7}

\bibitem{starexec}
Stump, A., Sutcliffe, G., Tinelli, C.: \tool{StarExec}: {A} cross-community
  infrastructure for logic solving. In: IJCAR~'14. pp. 367--373. LNCS 8562
  (2014). \doi{10.1007/978-3-319-08587-6\_28}

\bibitem{smartACE}
Wesley, S., Christakis, M., Navas, J.A., Trefler, R.J., W{\"{u}}stholz, V.,
  Gurfinkel, A.: Verifying \pl{Solidity} smart contracts via communication
  abstraction in \tool{SmartACE}. In: VMCAI~'22. pp. 425--449. LNCS 13182
  (2022). \doi{10.1007/978-3-030-94583-1\_21}

\bibitem{synthhorn}
Zhu, H., Magill, S., Jagannathan, S.: A data-driven {CHC} solver. In: {PLDI}
  '18. pp. 707--721 (2018). \doi{10.1145/3192366.3192416}

\end{thebibliography}
